\documentclass[10pt,a4paper,english]{article}\usepackage[]{graphicx}\usepackage[]{xcolor}
\makeatletter
\def\maxwidth{ %
  \ifdim\Gin@nat@width>\linewidth
    \linewidth
  \else
    \Gin@nat@width
  \fi
}
\makeatother

\definecolor{fgcolor}{rgb}{0.345, 0.345, 0.345}

\usepackage{framed}
\makeatletter
 {\par\unskip\endMakeFramed%
 \at@end@of@kframe}
\makeatother

\definecolor{shadecolor}{rgb}{.97, .97, .97}
\definecolor{messagecolor}{rgb}{0, 0, 0}
\definecolor{warningcolor}{rgb}{1, 0, 1}
\definecolor{errorcolor}{rgb}{1, 0, 0}
\newenvironment{knitrout}{}{} 

\usepackage{alltt}

\batchmode

\usepackage[hyphens]{url}
\usepackage{amsmath}
\usepackage{amsfonts}
\usepackage{amssymb}
\usepackage{amsthm}

\theoremstyle{plain}
\newtheorem{theorem}{Theorem}[section]

\newtheorem{corollary}[theorem]{Corollary}

\theoremstyle{definition}

\theoremstyle{remark}

\PassOptionsToPackage{hyphens}{url}\usepackage{hyperref}
\usepackage{hyperref}
\usepackage[square,numbers]{natbib}
\usepackage{orcidlink}
\usepackage{gitinfo2}

\usepackage[newitem,newenum,increaseonly]{paralist}


\usepackage[colorinlistoftodos]{todonotes}

\usepackage[notheorems]{common}
\usepackage{debsr}

    \usepackage{color} 
    \usepackage{enumerate} 
    \usepackage{amsmath} 
    \usepackage{amssymb} 
		\usepackage{fancyvrb} 
    \usepackage{hyperref}

    \definecolor{orange}{cmyk}{0,0.4,0.8,0.2}
    \definecolor{darkorange}{rgb}{.71,0.21,0.01}
    \definecolor{darkgreen}{rgb}{.12,.54,.11}
    \definecolor{myteal}{rgb}{.26, .44, .56}
    \definecolor{gray}{gray}{0.45}
    \definecolor{lightgray}{gray}{.95}
    \definecolor{mediumgray}{gray}{.8}
    \definecolor{inputbackground}{rgb}{.95, .95, .85}
    \definecolor{outputbackground}{rgb}{.95, .95, .95}
    \definecolor{traceback}{rgb}{1, .95, .95}
    \definecolor{red}{rgb}{.6,0,0}
    \definecolor{green}{rgb}{0,.65,0}
    \definecolor{brown}{rgb}{0.6,0.6,0}
    \definecolor{blue}{rgb}{0,.145,.698}
    \definecolor{purple}{rgb}{.698,.145,.698}
    \definecolor{cyan}{rgb}{0,.698,.698}
    \definecolor{lightgray}{gray}{0.5}
    
    \definecolor{darkgray}{gray}{0.25}
    \definecolor{lightred}{rgb}{1.0,0.39,0.28}
    \definecolor{lightgreen}{rgb}{0.48,0.99,0.0}
    \definecolor{lightblue}{rgb}{0.53,0.81,0.92}
    \definecolor{lightpurple}{rgb}{0.87,0.63,0.87}
    \definecolor{lightcyan}{rgb}{0.5,1.0,0.83}

    \definecolor{incolor}{rgb}{0.0, 0.0, 0.5}
    \definecolor{outcolor}{rgb}{0.545, 0.0, 0.0}

    \hypersetup{
      breaklinks=true,  
      colorlinks=true,
      urlcolor=blue,
      linkcolor=darkorange,
      citecolor=darkgreen,
      }
    
\IfFileExists{upquote.sty}{\usepackage{upquote}}{}
\begin{document}

\title{Post-Selection Estimation of Sharpe Ratios}
\author{\orcidlink{0000-0002-4197-6195} Steven E. Pav \thanks{\email{steven@gilgamath.com}.
}}
\date{\today}
\maketitle

\begin{abstract}
We consider the problem of estimating the true \txtSR of an asset selected for having the highest observed in-sample \txtSR among many assets.
We discuss estimators based on the polyhedral lemma, James Stein shrinkage, debiasing the expected maximum \txtSR, thresholding and empirical Bayes.
We test these estimators in simulations, computing bias and root mean square error across different values of sample size, number of assets, 
and spread and shape of population \txtSRs. 
We also compute rank correlation of the estimators against the underlying quantity, simulating how
these estimators might be used to compare or rank the output of different teams which perform this selection process.
We find that the James Stein estimator provides the best performance across many different realistic values of the relevant parameters,
followed by the GMLEB estimator of \citeauthor{jiang2009generalmaximumlikelihoodempirical}.
These results are fairly robust to correlation of asset returns, with some caveats.
\end{abstract}

\section{Introduction}

From the dawn of computing, quantitative-minded traders have used computers to try to analyze and predict market movements.  \cite{kaufman1978commodity,cootner1964random}
If the quality of outcomes was monotonic in computing power invested, however, there would be many rich nerds indeed.
This seems not to be the case, instead devising and backtesting thousands or millions of candidate trading strategies often results in disappointing performance out of sample.
There is a clear need even still for statistical tools to correct for the effects of selecting strategies based on their in-sample performance.

The setup we consider is fairly simple: a quantitative strategist devises and backtests some large number of strategies,
selects the one with the highest in-sample \txtSR, then wishes to estimate or perform inference on the population \txtSR of that strategy.
Early inferential approaches to the problem relied on corrections for familywise error rate.  \cite{bonferroni1936teoria,handbook2022}
\citeauthor{White:2000} proposed the use of the bootstrap for his ``Reality Check'',
which was improved upon by \citeauthor{Hansen:2005} \emph{inter alia}.  \cite{White:2000,Hansen:2005}
Previously we compared several procedures for the inferential task, including approaches based on directional alternatives
and a conditional test based on the ``polyhedral lemma'' which finds currency in the field of selective inference.  \cite{pav2019maxsharpe,nla.cat-vn3800977,lee2013exact}
More recent approaches leverage statistical learning theory for the inference problem. \cite{paleologo2025}

In this work we consider the estimation problem, namely:
conditional on selecting an asset or strategy because it has the highest in-sample \txtSR,
can we compute an estimate of the \txtSNR of that asset or strategy which is unbiased or has low mean square error?
Outside the context of quantitative trading, this problem has considerable history. 
One approach to the problem is via \emph{thresholding}, where values above a threshold (perhaps chosen empirically) are shrunk by a certain amount towards zero.
\cite{Donoho01121995,johnstone2004needles,JSSv012i08}
This can also be viewed from the viewpoint of Empirical Bayes.  \cite{04f9ad2d-909a-34da-ad09-b5ac91b665b6}
\citeauthor{reid2017post} introduced a maximum likelihood estimator based on the polyhedral lemma, and compared it to several other techniques,
including the threshold-based methods, as well as the James-Stein estimator.  \cite{reid2017post,james1961estimation}

This note does not seek to introduce new methods for the problem.
Rather we adapt existing methods to the problem of estimating the \txtSNR of the asset with maximum \txtSR,
and test them on population parameters which are likely to be of interest to practitioners,
with realistic sample sizes, effect sizes and distribution and number of strategies considered.
In that regard, this paper is quite similar to that of \citeauthor{reid2017post}, but adapted to the econometrics of the \txtSR. \cite{reid2017post}

\section{Estimation of the \txtSR}

The \txtSR is the most widely used statistic for evaluating and comparing trading strategies. \cite{Sharpe:1966,pav_the_book}
Defined as the sample mean of observed percent returns divided by the standard deviation of the same, the \txtSR is
exactly equal to the $t$-statistic up to scaling.
We wish to perform inference on what some would call the \emph{ex ante} \txtSR, what we call the \txtSNR,
which uses population mean and standard deviation.

We assume one has access to, or has constructed, \nstrat trading strategies.
Let 
$$
\psnr[i] = \frac{\pmu[i]}{\psig[i]}
$$
be the \txtSNR of the \kth{i} strategy. Collect all of them into the $\nstrat$-vector \pvsnr.
We backtest or observe \ssiz days of returns and compute the \txtSR of the \kth{i} strategy, call it $\ssr[i]$,
which we collect in the vector \svsr.

To simplify the exposition we assume that the indices have been reordered after the statistics are observed
so that
\begin{equation}
\ssr[1] \le \ssr[2] \le \ldots \le \ssr[\nstrat].
\end{equation}
We synchronize the indexing between the population and sample values,
so this does \emph{not} imply the same ordering holds on the $\psr[i]$.
That is $\psr[\nstrat]$ is the \txtSNR of the strategy which was observed to have the highest \txtSR, but it might not be the largest of the $\psr[i]$.
We simply prefer to use indexing like $\ssr[i]$ than $\ssr[(i)]$, which is more standard for this problem.

We suppose that the returns of the strategies are correlated with correlation matrix $\RMAT$.
When returns are elliptically distributed we showed previously that
\begin{equation}
\svsr \approx 
\normlaw{\pvsnr,\oneby{\ssiz}\wrapParens{\RMAT +
	\frac{\kurty-1}{4} \ogram{\pvsnr} + 
	\frac{\kurty}{2} \Mdiag{\pvsnr} \wrapParens{\RMAT\hadm\RMAT} \Mdiag{\pvsnr}}},
\label{eqn:apx_srdist_elliptical}
\end{equation}
where \kurty is the ``kurtosis factor'', equal to one third the kurtosis of the
marginals.  \cite{pav2019maxsharpe}
When measured on daily returns we expect the \txtSNR to be not much larger than around $0.16\dayto{-1/2}$,
corresponding to around $2.5\yrto{-\halff}$ assuming $252$ trading days per year.
As such quantities quadratic in $\pvsnr$ tend to be very small, and we often can get away with the approximation
\begin{equation}
\svsr \approx 
\normlaw{\pvsnr,\oneby{\ssiz}\RMAT}.
  \label{eqn:apx_srdist_rough}
\end{equation}

Thus our problem is fairly well approximated by the classical problem of estimation with normal noise.
In much of our testing we will assume independent returns, where $\RMAT=\eye$.
A perhaps more accurate model for many strategy searches is that of \emph{compound symmetry} where
$\RMAT = \wrapParens{1 - \rho}\eye + \rho \vone\trvone$; this is a correlation matrix with $\rho$ on the off-diagonals.
A correlation matrix of this form is convenient because it preserves order in this sense: we can transform our way back to 
the independent case while preserving identity of the maximum element. \cite{pav2019maxsharpe}
That is 
\begin{equation}
  \vect{z} = \sqrt{\ssiz}\ichol{\RMAT} \wrapParens{\svsr - \pvsnr} \approx \normlaw{\vzero,\eye},
  \label{eqn:bonf_zform_simple}
\end{equation}
where $\ichol{\RMAT}$ is the inverse of the (symmetric) square root of $\RMAT$.
Under the compound symmetric model, $\ichol{\RMAT}$ is order preserving:
if $\vect{a} = \ichol{\RMAT}\vect{b}$ and $b_i \le b_j$ then $a_i \le a_j$.
Thus the \kth{\nstrat} element of $\sqrt{\ssiz}\ichol{\RMAT}\svsr$ is the largest element
under our assumed ordering.
This is all to say we will focus mostly on the $\RMAT=\eye$ case in our exposition.

Given $\svsr$ which is approximately $\pvsnr$ plus a multivariate noise with covariance $\frac{1}{\ssiz}\eye$,
how can we estimate $\psnr[\nstrat]$ in a way that is unbiased or has small mean square error?
If all the elements of \pvsnr were the same, say equal to \psnr[0], we could avail ourselves of the result of 
\citeauthor{bailey2014pseudomathematics} on the expected value of the maximum of independent Gaussians.  \cite{bailey2014pseudomathematics}
In our formulation this can be expressed as 
\begin{equation}
  \Eof{\ssr[\nstrat]} \approx \psnr[0] + \ssiz^{-\halff}\wrapParens{\wrapParens{1-\mascher}\pinorm[1 - \frac{1}{\nstrat}] + \mascher\pinorm[1 - \frac{1}{\nstrat e}]},
\end{equation}
where $\mascher \approx 0.5772157$ is the Euler-Mascheroni constant.
Thus a nearly debiased estimate could be had by subtracting the second summand from $\ssr[\nstrat]$.
Of course, if all the elements of \pvsnr were equal we would not have to rely on the expected value of the extremum, 
we could just estimate \psnr[0] via the average (or ``grand mean'') of the elements of $\svsr$ defined as
$$
\ssravg = \frac{1}{\nstrat} \sum_i \ssr[i].
$$

We do not expect either of those estimators to be very good except in cases where
the variation in \pvsnr is small compared to the noise in \svsr, which is essentially $\ssiz^{-\halff}$.
Note that it is fairly common to collect backtests of length between around one and ten years, which corresponds
to $\ssiz$ varying from around 250 to around 2,500. 
Thus the standard deviation of noise in \svsr is somewhere around $0.02$ to 
$0.06$.

When the variation in \svsr is not small relative to $\ssiz^{-\halff}$, we need better estimators.
One simple classical approach is via the James-Stein estimator.  \cite{james1961estimation}
Suppose the elements of \pvsnr were randomly generated with variance $\psnrsigsq$,
then we observe $\svsr$ with variance $1/\ssiz$ on top of that. 
Then the unconditional variance of the $\ssr[i]$ would be their sum, $\psnrsigsq + 1/\ssiz$.
If you could rescale the elements of \svsr towards the (unknown) mean value of the $\pvsnr$ by a certain
factor, then the same $\svsr$ would have the same variance as the sample.
This handwavey explanation would lead one to the wrong scaling factor, however.
The actual James Stein estimator, which has lower mean square error than the usual least squares estimator (\ie \svsr itself in this case),
has one compute the shrinkage factor
$$
s = \wrapParens{1 - \frac{\wrapParens{\nstrat - 2}\oneby{\ssiz}}{\norm{\svsr}^2}}^{+},
$$
where $x^{+}$ is the maximum of $x$ and $0$, or the ``positive part'' operator.
One would then use $s \svsr$ as a better estimate of \pvsnr, in the least squares sense.
For our problem that means using $s \ssr[\nstrat]$ as an estimate of \pvsnr[\nstrat].
Rather than shrink to zero, we wish our estimator to be equivariant to locational shifts, so instead we shrink to the grand mean.
This means we compute
$$
s_v = \wrapParens{1 - \frac{\wrapParens{\nstrat - 2}\oneby{\ssiz}}{\norm{\svsr - \ssravg\vone}^2}}^{+},
$$
then use
$$
\ssr[JS] = \ssravg + s_v \wrapParens{\ssr[\nstrat] - \ssravg}
$$
to estimate $\psnr[\nstrat]$.

A number of other estimators that we consider can be expressed in a similar form.
These are based on the idea of ``thresholding''.
Let 
$$
\funcit{g}{x, t} \defeq \sign{x}\abs{x-t}^{+}
$$
be the thresholding function.
It shrinks anything smaller than $t$ in absolute value to zero, while values larger than that threshold are shifted towards zero by that amount.
Thresholding estimators are completely described by the way that they construct the threshold $t$.
That is, a thresholding estimator of \psnr[\nstrat] is computed as
$$
\ssravg + \funcit{g}{\ssr[\nstrat], \funcit{t}{\svsr - \ssravg\vone, \ssiz}},
$$
for some function $\funcit{t}{\cdot, \cdot}$.

The SURE estimator of \citeauthor{Donoho01121995} uses as $\funcit{t}{\cdot, \cdot}$ a result of Stein to estimate the 
threshold which would give the least squares best estimate of the entire vector. \cite{Donoho01121995}
For our problem, focused on only one element of the vector, this might not be optimal.
We envision exploring variants of this estimator in future revisions of this paper.

The Empirical Bayes estimate of \citeauthor{johnstone2004needles} can
also be expressed as a threshold estimator.  \cite{johnstone2004needles,JSSv012i08}
It starts from a prior distribution on the elements of $\pvsnr$ with a point mass at zero, and 
some distribution on the non-zero elements.
The observed data are then integrated, then the posterior median is used for the threshold.
\cite{reid2017post,04f9ad2d-909a-34da-ad09-b5ac91b665b6}

The GMLEB estimator of 
\citeauthor{jiang2009generalmaximumlikelihoodempirical}
is another threshold estimator.  \cite{jiang2009general}
This estimator is considerably harder to describe, and we refer the reader to the original paper.
We used the GMLEB code from \citeauthor{reid2017post} \cite{reid2017post}

Another approach to the problem is via the ``polyhedral lemma'' of \citeauthor{lee2013exact} \cite{lee2013exact}
Previously we applied this result to the problem of performing inference for this very same problem setup.  \cite{pav2019maxsharpe}
So one very simple estimator would be to compute the $0.5$ confidence bound using the truncated normal distribution
that arises from the problem setup.
For the simple case of $\RMAT=\eye$, this polyhedral lemma states that, conditional on 
$\ssr[1] \le \ssr[2] \le \ldots \le \ssr[\nstrat]$, the random variable
$$
u = \frac{\pnorm[\sqrt{\ssiz}\wrapParens{\ssr[\nstrat] - \psnr[\nstrat]}] - 
\pnorm[\sqrt{\ssiz}\wrapParens{\ssr[\nstrat-1] - \psnr[\nstrat]}]}{1 - 
\pnorm[\sqrt{\ssiz}\wrapParens{\ssr[\nstrat-1] - \psnr[\nstrat]}]}
$$
is uniformly distributed.
So the polyhedral median estimator is the value, found numerically, such that plugged in for the unknown $\psnr[\nstrat]$ gives a value of $0.5$ for $u$. 
That is we find $\psnr[0.5]$ to solve for 
$$
\half = \frac{\pnorm[\sqrt{\ssiz}\wrapParens{\ssr[\nstrat] - \psnr[0.5]}] - 
\pnorm[\sqrt{\ssiz}\wrapParens{\ssr[\nstrat-1] - \psnr[0.5]}]}{1 - 
\pnorm[\sqrt{\ssiz}\wrapParens{\ssr[\nstrat-1] - \psnr[0.5]}]}.
$$

\citeauthor{reid2017post} also start from the polyhedral lemma, but use it to compute a maximum likelihood estimator for \psnr[\nstrat].  \cite{reid2017post}
The starting point is that, conditional on the observed ordering, the variable $\sqrt{\ssiz}\wrapParens{\ssr[\nstrat] - \psnr[\nstrat]}$
follows a truncated normal distribution and has probability density function:
$$
\funcit{f}{x; \psnr[\nstrat]} = \sqrt{\ssiz}\frac{\indic{\sqrt{\ssiz}\wrapParens{\ssr[\nstrat-1] - \psnr[\nstrat]} \le x}\dnorm[x]}{1 - \pnorm[\sqrt{\ssiz}\wrapParens{\ssr[\nstrat-1] - \psnr[\nstrat]}]}.
$$
Given just the single observation $\ssr[\nstrat]$, we find, numerically, the value of \psnr[\nstrat] that maximizes this likelihood.
That is we compute the MLE as
$$
\argmax_{\psnr} \funcit{f}{\sqrt{\ssiz}\wrapParens{\ssr[\nstrat] - \psnr}; \psnr}.
$$
Pace \citeauthor{reid2017post} we found this estimator to be quite unstable in our simulations.
One distinction is that \citeauthor{reid2017post} considered the problem of estimating the $m$ largest elements of a noisy vector in absolute value,
which is somewhat different from our problem and perhaps less sensitive to small values of $\ssr[\nstrat] - \ssr[\nstrat-1]$.




\section{Simulations}

We perform a number of Monte Carlo simulations of the strategy backtest-select-estimate pipeline we describe above.
We evaluate the different estimators in two different ways:
in one we simply compute the bias and root mean square error (RMSE) of the estimator compared to \psnr[\nstrat];
in the other we simulate using one of these estimators to select from several different sources of strategies, 
and compute the correlation of the estimators against the \psnr[\nstrat].

At its heart a single simulation consists of the following:
\begin{compactenum}
\item Either the \pvsnr is given, or it is generated according to some distribution; when randomly generated, we usually control the \psnrsigsq.
\item We then generate \ssiz days of independent normally distributed returns for \nstrat assets with zero mean and unit variance.
\item To each column we add the value of \pvsnr so each column is a $\ssiz$-vector with \txtSNR equal to the corresponding element of \pvsnr.
\item We compute the \txtSR of each column.
\item We apply the various estimators to the vector \svsr and record them. We also record the \psnr[\nstrat].
\end{compactenum}
We typically perform many simulations for each setting of the \pvsnr, \ssiz and \nstrat.
Over those many simulations we compute the bias and RMSE of the estimator compared to the \txtSNR of the strategy selected for having the highest \txtSR.
We then plot the bias or RMSE versus, say, the \psnrsigsq, or versus \nstrat or \ssiz and so on.

\paragraph{Forms of the population vector}

We consider a number of different layouts for the vector \pvsnr.
Some of these are random and the \pvsnr is generated afresh for each simulation.
In others the vector \pvsnr is fixed at one value across many simulations.
The layouts are:

\begin{compactitem}
\item Gaussian: The \pvsnr is drawn from a Gaussian distribution with mean zero and variance \psnrsigsq.
\item Uniform: The \pvsnr is drawn from a continuous uniform distribution with mean zero and variance \psnrsigsq. That is, elements of \pvsnr are drawn uniformly from
  $-\sqrt{3}\psnrsig$ to $\sqrt{3}\psnrsig$.
\item Bimodal: The \pvsnr takes value $\pm \psnrsig$ with probability \half.
\item all-good: The \pvsnr equals $\psnr[0]\vone$ for some $\psnr[0]$. This is a non-random configuration, constant across all simulations with this layout.
  We test this layout with various values of $\psnr[0]$ mainly as a check that our estimators are location invariant.
\item one-good: One element of \pvsnr is \psnr[0] while the rest are $-\psnr[0]$ for some $\psnr[0]$. This is a non-random configuration.
  When $\psnr[0]$ is relatively large compared to $\ssiz^{-\halff}$ we are likely to select the ``good'' strategy, but some of our estimators are likely to be very biased.
\item two-good: Two elements of \pvsnr are \psnr[0] while the rest are $-\psnr[0]$ for some $\psnr[0]$. This is a non-random configuration. 
\end{compactitem}

\paragraph{Estimators}

We consider the following estimators:

\begin{compactitem}
\item Biased: This is \ssr[\nstrat], the maximal \txtSR of all those observed. We include this as a benchmark for the alternative estimators.
\item Grand Mean: This is \ssravg, the average of all the observed \txtSR values.
\item James-Stein: The James-Stein estimator, shrunk towards the grand mean.
\item GMLEB: The estimator of \citeauthor{jiang2009generalmaximumlikelihoodempirical} shrunk to the grand mean.
\item SURE: The estimator of \citeauthor{Donoho01121995} shrunk to the grand mean.
\item Polyhedral Median: This is the 0.5 confidence bound based on the polyhedral lemma.
\item Polyhedral MLE: The MLE estimator of \citeauthor{reid2017post}
\item Expected Max: This is \ssr[\nstrat] minus the expected value of the maximum found by \citeauthor{bailey2014pseudomathematics}
\item Empirical Bayes: The estimator of \citeauthor{johnstone2004needles}, shrunk to the grand mean.
We use the code from the \Rpackage{EbayesThresh} \Rlang package. \cite{EbayesThresh}
\end{compactitem}

\subsection{Bias and Square Error Results}

First we performed simulations using the Gaussian, Uniform and Bimodal distributions for \pvsnr.
In these we let \psnrsig range from 
$0\yrto{-\halff}$ to 
$1\yrto{-\halff}$.
We let \nstrat take values 10, 100, 1000.
We varied \ssiz to represent between 
$2$ and
$8$ years at $252$ days per year.
For each setting we perform $500$ simulations.
In \figref{sim_rnd_plotz_bias} we plot the bias versus $\psnrsig$,
and in \figref{sim_rnd_plotz_rmse} we plot the RMSE versus $\psnrsig$,
fixing \ssiz at 4 years, and plotting for the various values of \nstrat tested.

Based on early results of these studies, we removed the Polyhedral Median and MLE estimators from the bias and RMSE plots.
Both consistently exhibited high RMSE, much higher than the other estimators.
The issue arises when the $\ssr[\nstrat]$ and $\ssr[\nstrat-1]$ are very close to each other;
in this case the truncated normal that arises from the polyhedral lemma is very close to its limiting
value which causes very extreme values of the estimators, see \secref{polyhedral_lemma} for more details.
For small numbers of simulations you might not hit such a case, but we did when performing $500$ simulations.
We will consider these two estimators in the following subsection, but will see they perform poorly there as well.

In those plots we see that the Biased estimator is indeed biased, and often has the highest RMSE.
The Grand Mean is unbiased and efficient when $\psnrsig=0$, but otherwise negatively biased and has high RMSE.
The James Stein estimator is, to the resolution of the plot, effectively unbiased for the Gaussian layout and has the lowest RMSE in that case as well;
when \pvsnr is uniform or bimodal, James Stein exhibits some positive bias and increased RMSE for larger \psnrsig.
GMLEB is slightly positively biased for uniform and bimodal \pvsnr, but appears to have RMSE which is flat with respect to \pvsnr.
The Expected Max estimator exhibits negative bias for Gaussian \pvsnr, but is better behaved for uniform and bimodal layouts.

\begin{knitrout}\small
\definecolor{shadecolor}{rgb}{0.969, 0.969, 0.969}\color{fgcolor}\begin{figure}[ht]
\includegraphics[width=0.975\textwidth,height=0.691\textwidth]{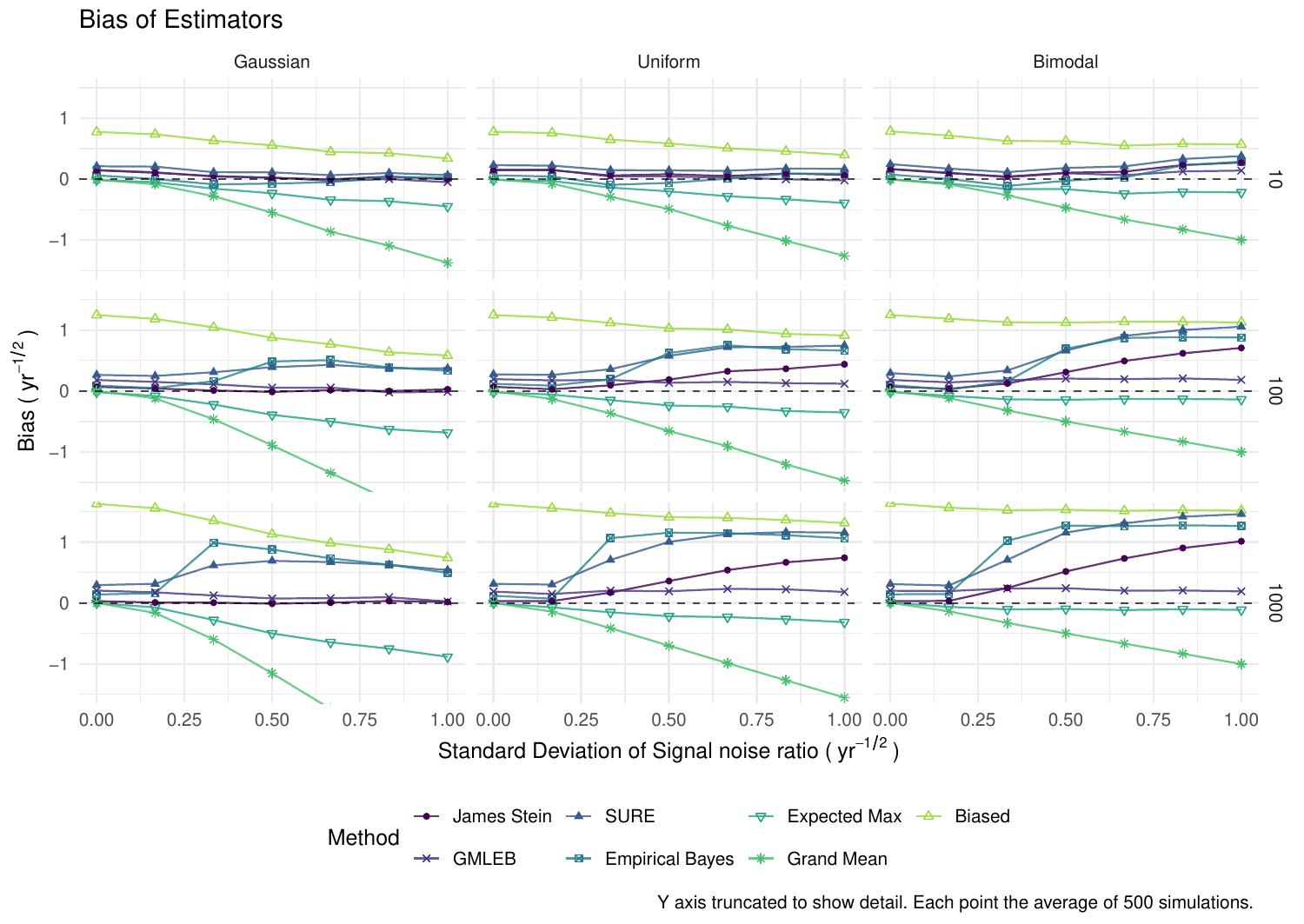} \caption[The empirical biases of the tested estimators are shown versus \psnrsig, the spread of the \txtSNR values]{The empirical biases of the tested estimators are shown versus \psnrsig, the spread of the \txtSNR values. Facet columns show the three different random layouts of the \pvsnr values, facet rows show the number of strategies, \nstrat. All simulations tested under $\ssiz=1008$.}\label{fig:sim_rnd_plotz_bias}
\end{figure}

\end{knitrout}

\begin{knitrout}\small
\definecolor{shadecolor}{rgb}{0.969, 0.969, 0.969}\color{fgcolor}\begin{figure}[ht]
\includegraphics[width=0.975\textwidth,height=0.691\textwidth]{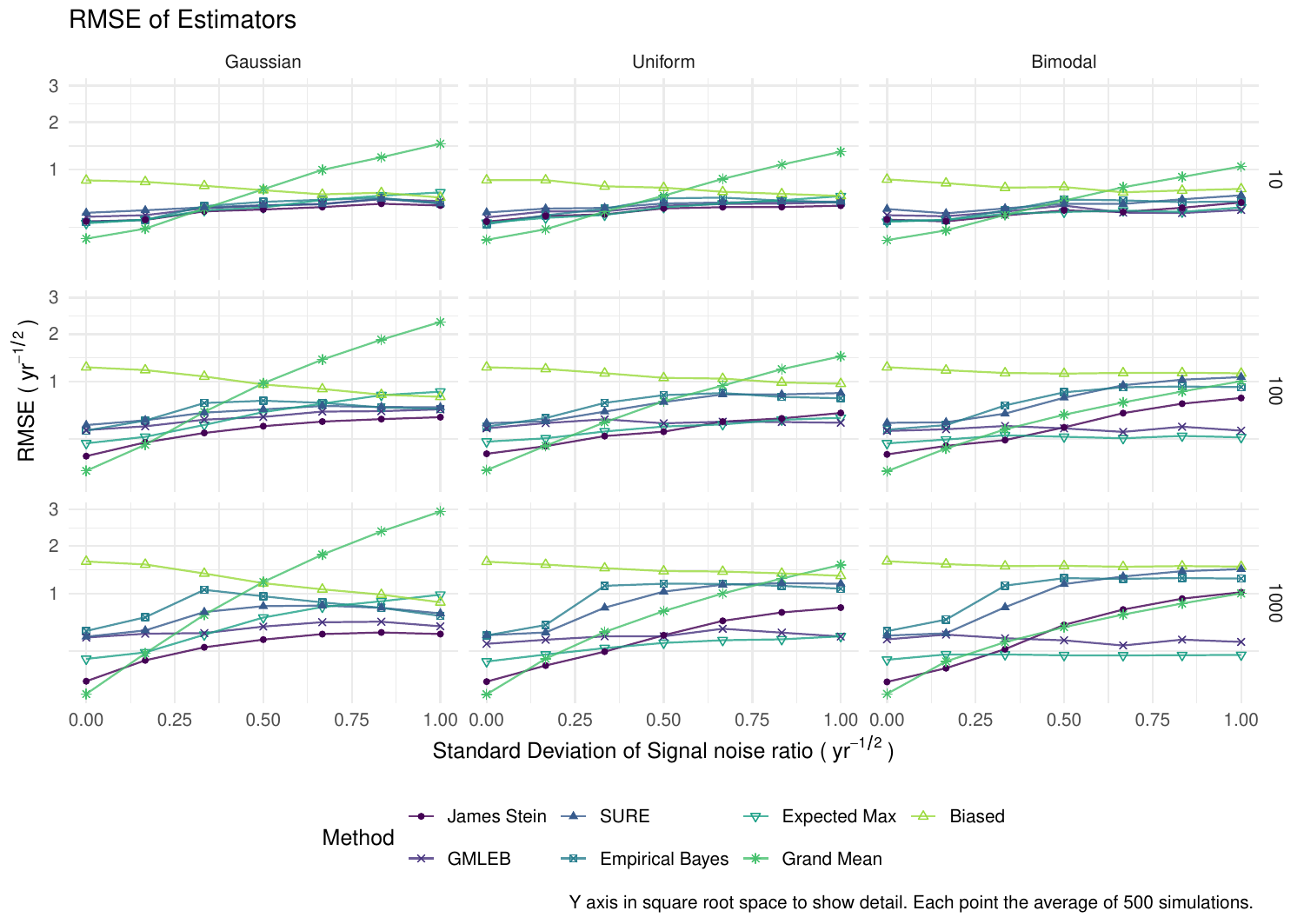} \caption[The empirical RMSE values of the tested estimators are shown versus \psnrsig, the spread of the \txtSNR values]{The empirical RMSE values of the tested estimators are shown versus \psnrsig, the spread of the \txtSNR values. Facet columns show the three different random layouts of the \pvsnr values, facet rows show the number of strategies, \nstrat. All simulations tested under $\ssiz=1008$.}\label{fig:sim_rnd_plotz_rmse}
\end{figure}

\end{knitrout}

In \figref{sim_rnd_plotz_bias_II} we plot the bias versus $\psnrsig$,
and in \figref{sim_rnd_plotz_rmse_II} we plot the RMSE versus $\psnrsig$
for these simulations, but fixing $\nstrat=100$ and varying \ssiz in the different facet rows.
This shows how error depends on the sample size when the number of strategies is fixed.
The middle row of \figref{sim_rnd_plotz_bias_II} should be the same as the middle row of \figref{sim_rnd_plotz_bias},
and similarly for \figref{sim_rnd_plotz_rmse_II} and \figref{sim_rnd_plotz_rmse}.
One thing to observe in these plots is how the estimators respond to varying \ssiz:
the Grand Mean seems almost unaffected, which makes sense; most of the other estimators
appear to have less bias for larger \ssiz, but this is tricky to quantify.

In \figref{sim_rnd_plotz_rmse_IIb} we reverse this relationship and plot the RMSE versus \ssiz in years
with lines for selected values of \psnrsig. 
We can see more clearly here that all of the methods seem to improve (have lower RMSE) with increasing
sample size except for the Grand Mean, and perhaps the James Stein method for bimodal layout and large \psnrsig.

\begin{knitrout}\small
\definecolor{shadecolor}{rgb}{0.969, 0.969, 0.969}\color{fgcolor}\begin{figure}[ht]
\includegraphics[width=0.975\textwidth,height=0.691\textwidth]{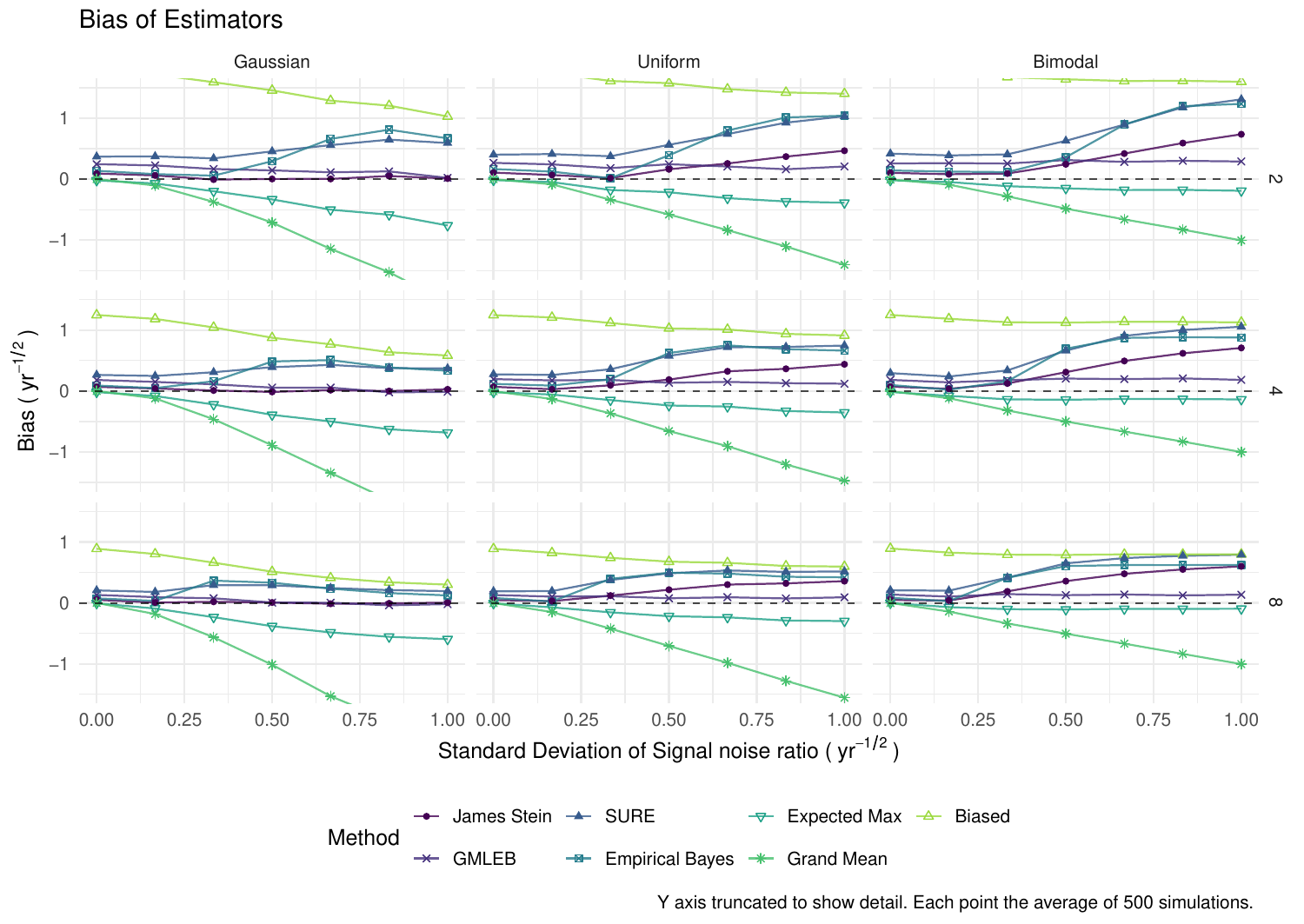} \caption[The empirical biases of the tested estimators are shown versus \psnrsig, the spread of the \txtSNR values]{The empirical biases of the tested estimators are shown versus \psnrsig, the spread of the \txtSNR values. Facet columns show the three different random layouts of the \pvsnr values, facet rows show the sample size, \ssiz, in years, assuming 252 days per year. All simulations tested under $\nstrat=100$.}\label{fig:sim_rnd_plotz_bias_II}
\end{figure}

\end{knitrout}

\begin{knitrout}\small
\definecolor{shadecolor}{rgb}{0.969, 0.969, 0.969}\color{fgcolor}\begin{figure}[ht]
\includegraphics[width=0.975\textwidth,height=0.691\textwidth]{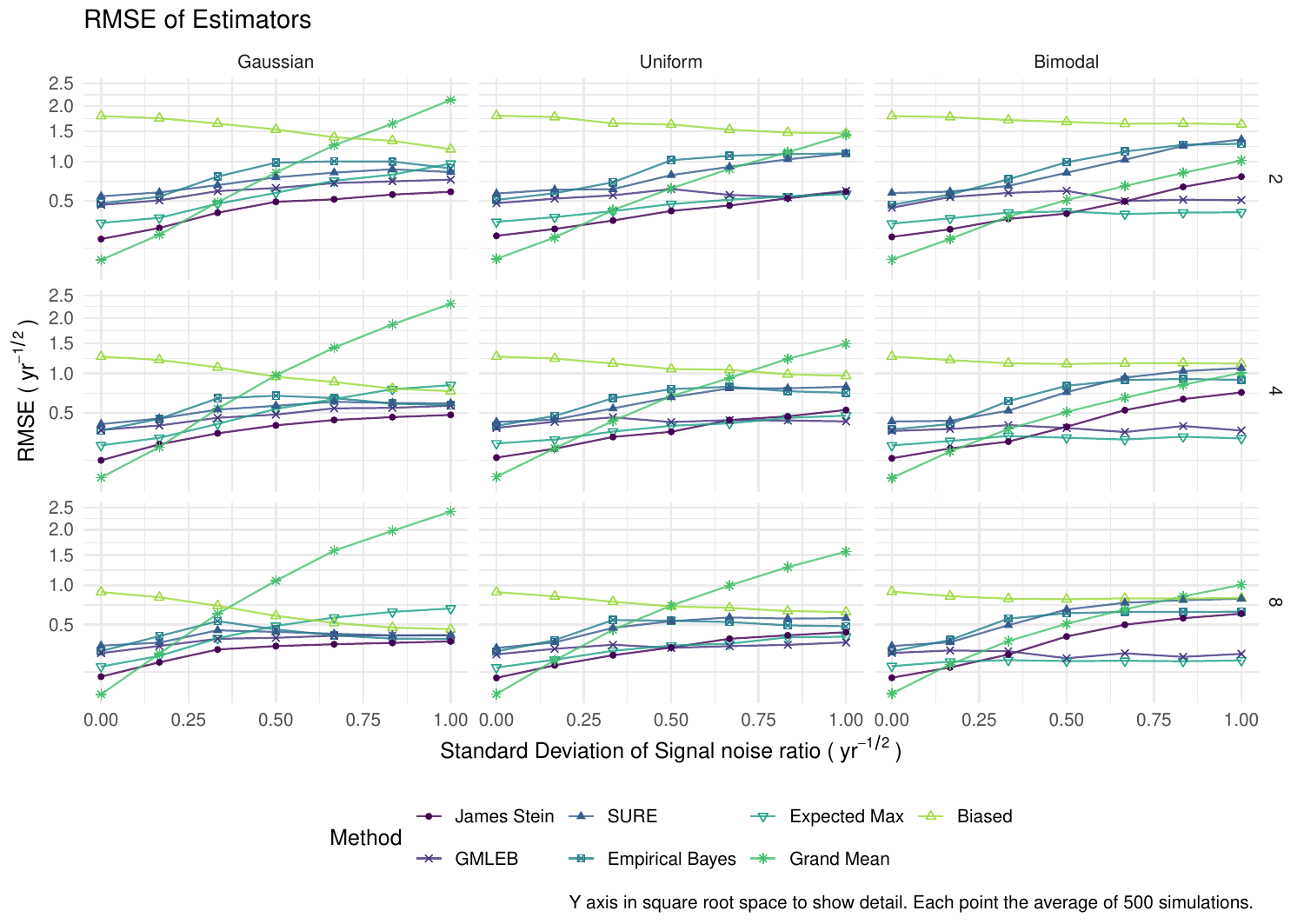} \caption[The empirical RMSE values of the tested estimators are shown versus \psnrsig, the spread of the \txtSNR values]{The empirical RMSE values of the tested estimators are shown versus \psnrsig, the spread of the \txtSNR values. Facet columns show the three different random layouts of the \pvsnr values, facet rows show the sample size, \ssiz, in years, assuming 252 days per year. All simulations tested under $\nstrat=100$.}\label{fig:sim_rnd_plotz_rmse_II}
\end{figure}

\end{knitrout}

\begin{knitrout}\small
\definecolor{shadecolor}{rgb}{0.969, 0.969, 0.969}\color{fgcolor}\begin{figure}[ht]
\includegraphics[width=0.975\textwidth,height=0.691\textwidth]{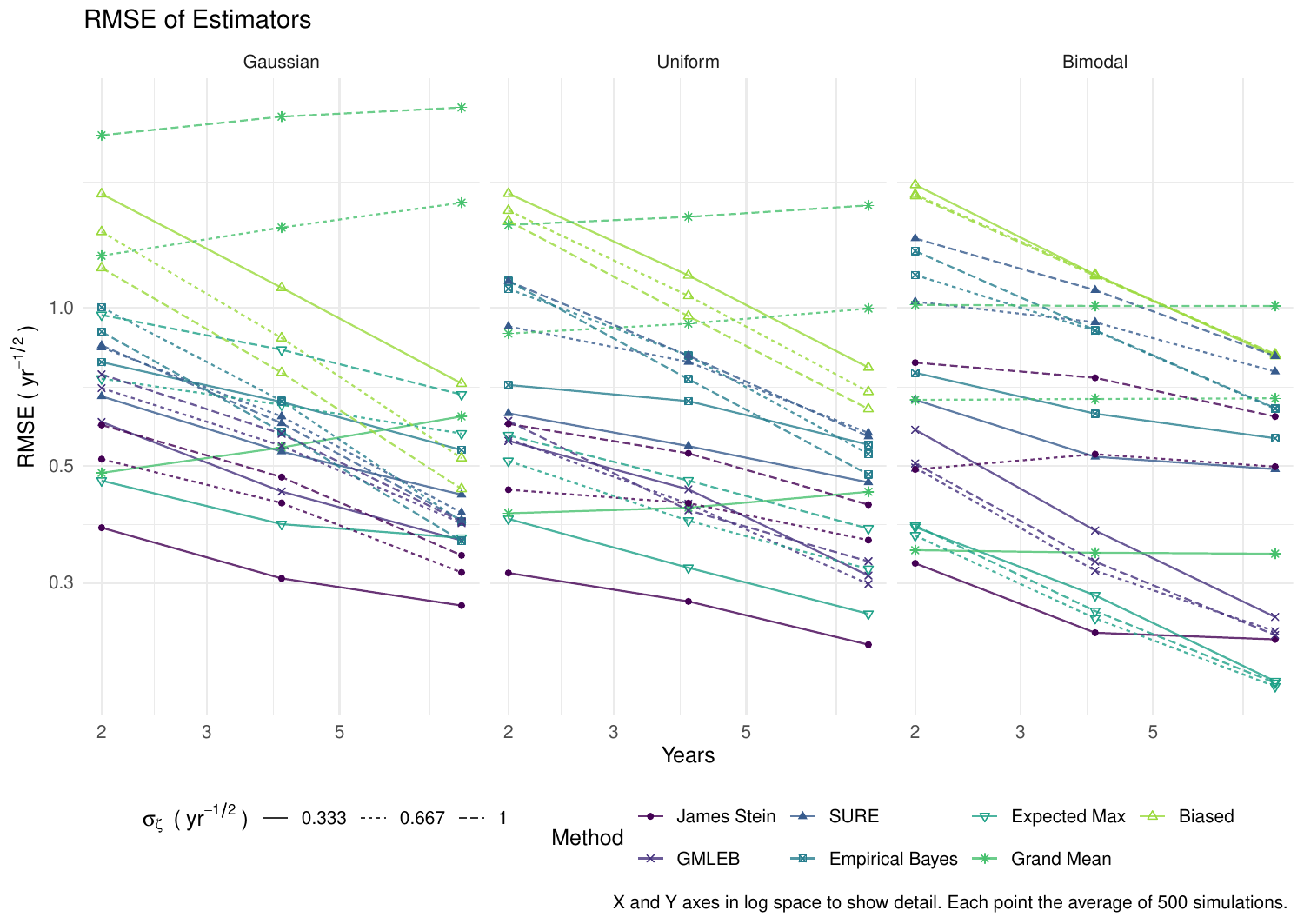} \caption[The empirical RMSE values of the tested estimators are shown versus \ssiz, in years, assuming 252 days per year]{The empirical RMSE values of the tested estimators are shown versus \ssiz, in years, assuming 252 days per year. Facet columns show the three different random layouts of the \pvsnr values. All simulations tested under $\nstrat=100$. Lines are shown for selected values of \psnrsig. }\label{fig:sim_rnd_plotz_rmse_IIb}
\end{figure}

\end{knitrout}

How can we make sense of the various plots above?
One way to summarize them is via the \emph{performance plots} of 
\citeauthor{dolan2004benchmarkingoptimizationsoftwareperformance}
\cite{dolan2004benchmarkingoptimizationsoftwareperformance}
We collect the simulations from above, then for a given choice of \ssiz, \nstrat, \psnrsig
and the layout distribution, we compute the ratio of the RMSE of an estimator
to the minimum RMSE over all estimators at that setting of the parameters.
We want this ratio to be as close to 1.0 as possible.
In the performance plot, in \figref{sim_plot_perf},
we plot the empirical CDF of this ratio for each method, with facets for the
different layout distributions.
In a performance plot, you wish to select a method which is ``up and to the left'' in the plot.
In our case James Stein clearly dominates in the Gaussian layout, is largely tied with Expected Max or slightly better
for Uniform layout, and somewhat worse than Expected Max and GMLEB for Bimodal layout.
A number of the methods show high regret in the Bimodal case, including SURE and Empirical Bayes.

We note that it is hard to take these results as conclusive,
since we only test a handful of different configurations of the relevant parameters,
and those might not be sampled in a way that is representative of what one expects in reality.

\begin{knitrout}\small
\definecolor{shadecolor}{rgb}{0.969, 0.969, 0.969}\color{fgcolor}\begin{figure}[ht]
\includegraphics[width=0.975\textwidth,height=0.487\textwidth]{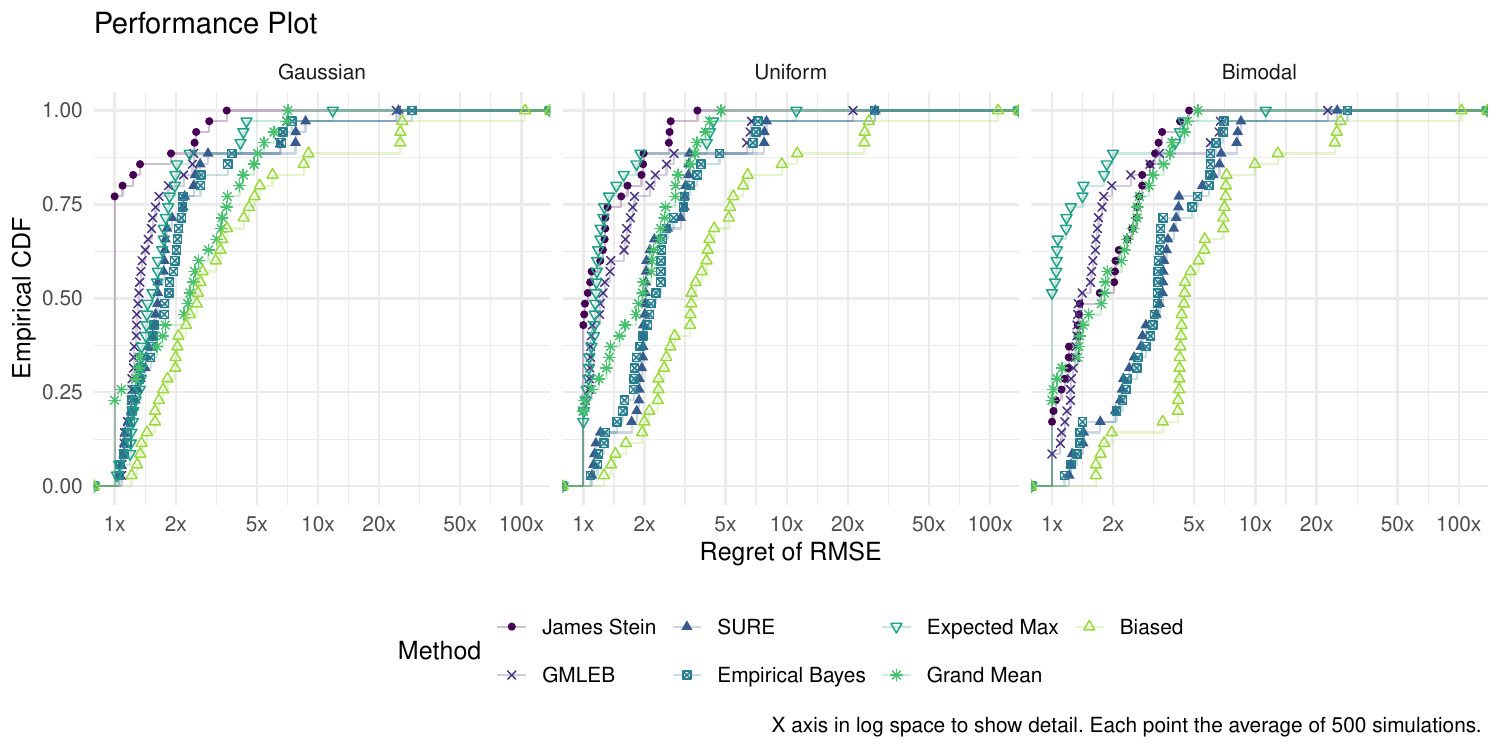} \caption[Empirical CDFs of the regret of each method are plotted in these performance plots]{Empirical CDFs of the regret of each method are plotted in these performance plots. Regret is the ratio of the RMSE of a method to the minimum RMSE of all methods for that parameter setting. }\label{fig:sim_plot_perf}
\end{figure}

\end{knitrout}

\clearpage

In \figref{sim_plotz_bias} we plot the bias versus $\psnr[0]$,
and in \figref{sim_plotz_rmse} we plot the RMSE versus $\psnr[0]$,
where \psnr[0] is the \txtSNR of the good strategy, for the all-good, one-good, and two-good layouts of \pvsnr.
In these we fix \ssiz at 4 years, and $\nstrat=100$.

For the all-good case we see, as hoped, that the bias and RMSE are essentially flat with respect to \psnr[0].
This means that our estimators are equivariant with respect to location shifts, as desired.
We also see that Grand Mean has the lowest RMSE, followed by James Stein, then Expected Max.
This is somewhat surprising as James Stein exhibits some positive bias in the all-good case,
as seen in \figref{sim_plotz_bias}, yet still exhibits lower RMSE than the less biased
Expected Max estimator, which must have a larger variance.
In the all-good case, the James Stein estimator will sometimes shrink all the way to the Grand Mean,
but sometimes will only shrink part of the way, which explains the observed positive bias.

All-good is an unlikely configuration for \pvsnr, but it is the limiting case where the variation of \pvsnr is small compared to $\ssiz^{-\halff}$.
The one-good and two-good cases are also somewhat unlikely, especially for larger \psnr[0].
For the one-good layout and large \psnr[0], the Biased estimator has lowest RMSE and is nearly unbiased, since effectively the
bad strategies have very little chance of being selected.
We can see the change in behavior clearly in the bias plots: after some point most of the estimators seem to 
reach an asymptotic value, except for Grand Mean, which gets progressively worse with \psnr[0].
The GMLEB appears to have low regret across \psnr[0] for the one-good and two-good cases.
If somehow one knew there were very few good strategies and a large number of bad ones,
GMLEB would be the recommended estimator, unless the gap between good and bad were known to be very large,
in which case the simple Biased estimator is to be used with caution.

\begin{knitrout}\small
\definecolor{shadecolor}{rgb}{0.969, 0.969, 0.969}\color{fgcolor}\begin{figure}[ht]
\includegraphics[width=0.975\textwidth,height=0.691\textwidth]{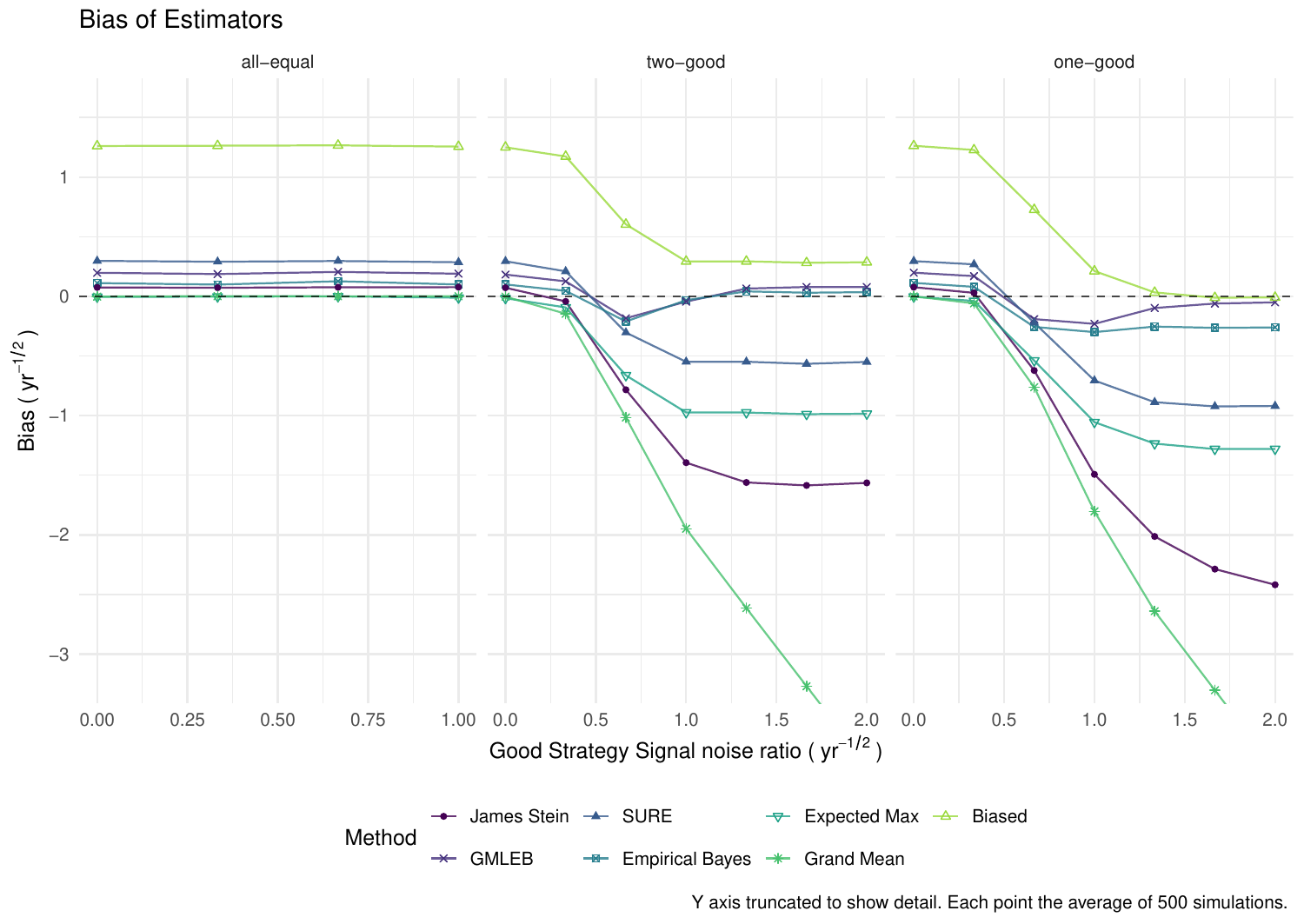} \caption[The empirical biases of the tested estimators are shown versus the \txtSNR of the good strategy for the fixed layout \pvsnr simulations]{The empirical biases of the tested estimators are shown versus the \txtSNR of the good strategy for the fixed layout \pvsnr simulations. Facet columns show the three different layouts of the \pvsnr values. All simulations tested under $\ssiz=1008$, and $\nstrat=100$. }\label{fig:sim_plotz_bias}
\end{figure}

\end{knitrout}

\begin{knitrout}\small
\definecolor{shadecolor}{rgb}{0.969, 0.969, 0.969}\color{fgcolor}\begin{figure}[ht]
\includegraphics[width=0.975\textwidth,height=0.691\textwidth]{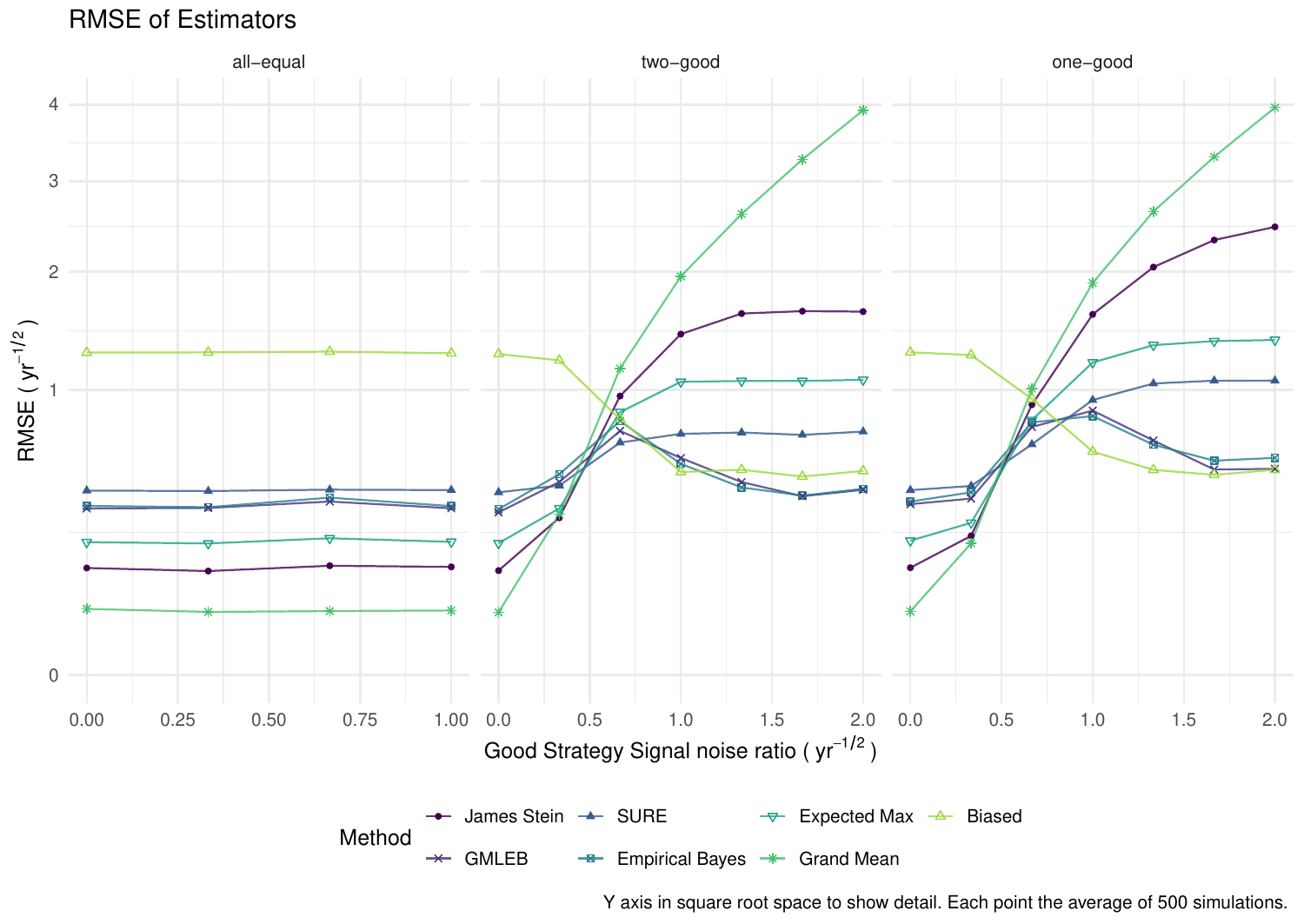} \caption[The empirical RMSE values of the tested estimators are shown versus the \txtSNR of the good strategy for the fixed layout \pvsnr simulations]{The empirical RMSE values of the tested estimators are shown versus the \txtSNR of the good strategy for the fixed layout \pvsnr simulations. Facet columns show the three different layouts of the \pvsnr values. All simulations tested under $\ssiz=1008$, and $\nstrat=100$. }\label{fig:sim_plotz_rmse}
\end{figure}

\end{knitrout}

\clearpage

\subsection{Ranking Results}

The simulations above looked at the quality of estimators \emph{qua} estimators of \psnr[\nstrat].
We wish to also examine how these estimators can be used to compare two or more selected strategies, 
each found by the process of selecting based on \txtSR.

As an illustration, 
suppose that former cryptographers Alice and Bob are now developing quantitative strategies.
Alice is \emph{adventurous}: she has many different ideas and performs many backtests.
Bob is \emph{boring}: he tests a few variants on one idea.
Both follow our process of generating some ideas, backtesting them, then selecting the one with the highest in-sample \txtSR.
Depending on the mean \txtSNR of their idea generating processes,
their \psnrsigsq, the \ssiz and \nstrat, either Alice or Bob might select a better strategy.
We wish an estimator computed on their respective \svsr vectors to be directionally correct for the \txtSNR of each of their chosen
strategies.
That is to say a good estimator would be one which is likely to identify the participant with higher \psnr[\nstrat].

To measure the association of the estimator and the population value \psnr[\nstrat], we perform a number of simulations across different configurations.
We then compute Kendall's rank correlation coefficient, $\tau$ of an estimate against the population value.
This gives some idea of the agreement of pairwise rankings of the two quantities.
We also compute Spearman's rank correlation coefficient, $\rho$ of an estimate against the population value.
This gives some idea of how likely an estimator is to correctly order a bunch of different realizations of Alice and Bob.

To test this use of estimators, we perform a bunch of simulations.
We create 500 ``corners'', where 
we pick $\psnrsig$ uniformly from 
$0\yrto{-\halff}$ to
$1\yrto{-\halff}$
and 
we select \nstrat log uniformly from 
$10$ to
$2,500$.
We cross this with \ssiz taking values from 
$0.5$ to $8$ years at $252$ days per year,
and cross this with the three different layouts, Gaussian, Uniform and Bimodal.

At each setting of the parameters we perform 8 simulations, mostly to take advantage of core parallelism.
We compute each estimator and the \psnr[\nstrat] in each simulation.
This represents a total of 
60,000 simulations.
We then compute the rank correlation coefficients of each estimator to the ground truth,
sometimes grouping by the layout or \ssiz, \nstrat, \psnrsig.

In \tabref{anb_1}, we tabulate the rank correlation coefficients for all estimators across all layouts, \ssiz, \nstrat, \psnrsig tested,
a total of 60,000 simulations.
Similarly, in \tabref{anb_2}, we tabulate the correlations grouped by layout for the three layouts, across \ssiz, \nstrat, \psnrsig. Each row is based on 20,000 simulations.
We bold the maximum value in each column, or each column and group, and methods are given in decreasing order of computed Kendall's correlation coefficient.

\begin{table}[ht]
\centering
\begin{tabular}{l|ll}
  \hline
Estimator & Kendall & Spearman \\ 
  \hline
James Stein & \textbf{0.57} & \textbf{0.76} \\ 
  GMLEB & 0.49 & 0.66 \\ 
  SURE & 0.46 & 0.64 \\ 
  Expected Max & 0.45 & 0.62 \\ 
  Empirical Bayes & 0.42 & 0.60 \\ 
  Polyhedral Median & 0.16 & 0.23 \\ 
  Polyhedral MLE & 0.15 & 0.23 \\ 
  Biased & 0.15 & 0.21 \\ 
  Grand Mean & 0.02 & 0.02 \\ 
   \hline
\end{tabular}
\caption{Rank correlation coefficients of the various tested estimators versus the \psnr[\nstrat] are shown. Rank correlations are computed across all layouts, \psnrsig, \ssiz and \nstrat. Each estimator tested on 60,000 simulations. } 
\label{tab:anb_1}
\end{table}

\begin{table}[ht]
\centering
\begin{tabular}{ll|ll}
  \hline
Estimator & Layout & Kendall & Spearman \\ 
  \hline
James Stein & Gaussian & \textbf{0.59} & \textbf{0.78} \\ 
  GMLEB & Gaussian & 0.51 & 0.69 \\ 
  Expected Max & Gaussian & 0.50 & 0.69 \\ 
  SURE & Gaussian & 0.49 & 0.68 \\ 
  Empirical Bayes & Gaussian & 0.47 & 0.65 \\ 
  Polyhedral Median & Gaussian & 0.23 & 0.33 \\ 
  Polyhedral MLE & Gaussian & 0.23 & 0.33 \\ 
  Biased & Gaussian & 0.18 & 0.26 \\ 
  Grand Mean & Gaussian & 0.03 & 0.04 \\ 
   \hline
James Stein & Uniform & \textbf{0.58} & \textbf{0.76} \\ 
  GMLEB & Uniform & 0.47 & 0.64 \\ 
  SURE & Uniform & 0.45 & 0.63 \\ 
  Expected Max & Uniform & 0.45 & 0.62 \\ 
  Empirical Bayes & Uniform & 0.41 & 0.59 \\ 
  Polyhedral Median & Uniform & 0.14 & 0.20 \\ 
  Polyhedral MLE & Uniform & 0.14 & 0.20 \\ 
  Biased & Uniform & 0.13 & 0.18 \\ 
  Grand Mean & Uniform & 0.01 & 0.02 \\ 
   \hline
James Stein & Bimodal & \textbf{0.62} & \textbf{0.8} \\ 
  GMLEB & Bimodal & 0.49 & 0.65 \\ 
  SURE & Bimodal & 0.46 & 0.64 \\ 
  Expected Max & Bimodal & 0.40 & 0.57 \\ 
  Empirical Bayes & Bimodal & 0.38 & 0.55 \\ 
  Biased & Bimodal & 0.14 & 0.20 \\ 
  Polyhedral Median & Bimodal & 0.09 & 0.13 \\ 
  Polyhedral MLE & Bimodal & 0.08 & 0.11 \\ 
  Grand Mean & Bimodal & 0.00 & 0.00 \\ 
   \hline
\end{tabular}
\caption{Rank correlation coefficients of the various tested estimators versus the \psnr[\nstrat] are shown. Rank correlations are computed grouped by layout, across all \psnrsig, \ssiz and \nstrat. Each estimator tested on 20,000 simulations. } 
\label{tab:anb_2}
\end{table}

From these simulations, 
we also plotted the empirical correlation of the estimators to \psnr[\nstrat] versus the dimensions of
\ssiz, \nstrat and \psnrsig.
In \figref{anb_plot_III} we plot the correlations against \ssiz, across all layouts and values of \nstrat and \psnrsig.
The data for this plot are tabulated in \tabref{anb_3} in \secref{additional_tables_and_figures} in the appendix.
Points are based on 12,000 simulations.
In \figref{anb_plot_IV} we plot the correlations against ranges of \nstrat, across all layouts and values of \ssiz and \psnrsig;
the companion data are in \tabref{anb_4}.
Points are based on from 10,440 to 14,400 simulations.
In \figref{anb_plot_V} we plot the correlations against ranges of \psnrsig, across all layouts and values of \ssiz and \nstrat;
the companion data are in \tabref{anb_5}.
Points are based on from 14,160 to 15,360 simulations.


In these plots and tables we see that the James Stein estimator dominates the other estimators, for both Kendall's and Spearman's correlation coefficients,
for the overall results and in nearly every subgroup we consider here.
The only exception to this is the large \psnrsig group
in \tabref{anb_5}, where James Stein is bested by GMLEB and Expected Max. 
Looking at the RMSE values for larger \psnrsig and non-Gaussian layouts in 
\figref{sim_rnd_plotz_rmse}, this is not a surprising finding.
GMLEB is consistently highly ranked in these tables as well,
typically taking the second place, though sometimes trading off with SURE or, rarely, the Expected Max.

We note that when \psnrsig is small, 
all the rank correlation coefficients are fairly low,
as seen in \figref{anb_plot_V}.
When the spread in skill is low, selection is heavily loaded on luck, and there is not much we can expect from these estimators.
We would point out, however, that rank correlation coefficients of estimators seem to increase in \ssiz,
as seen in \figref{anb_plot_III}. 
If one is stuck with a process that has low \psnrsig, increasing sample size would seem to improve the discriminating power of the
estimators, as well as increase the probability of picking a good strategy.
The rank correlation also seems to be generally increasing with \nstrat,
as seen in \figref{anb_plot_IV}.
We had expected, to the contrary, that larger \nstrat would present challenges to the estimators.
Perhaps this effect is driven by an increased spread in achieved \txtSNR.

Another factor evident in comparing \figref{anb_plot_V} to the other plots here is the overall lower rank correlation
coefficients when \psnrsig is fixed.
It would seem a fair amount of the discriminating power of these estimators seen in the other plots comes from
recognizing the overall higher spread in the \svsr values.

\begin{knitrout}\small
\definecolor{shadecolor}{rgb}{0.969, 0.969, 0.969}\color{fgcolor}\begin{figure}[ht]
\includegraphics[width=0.975\textwidth,height=0.691\textwidth]{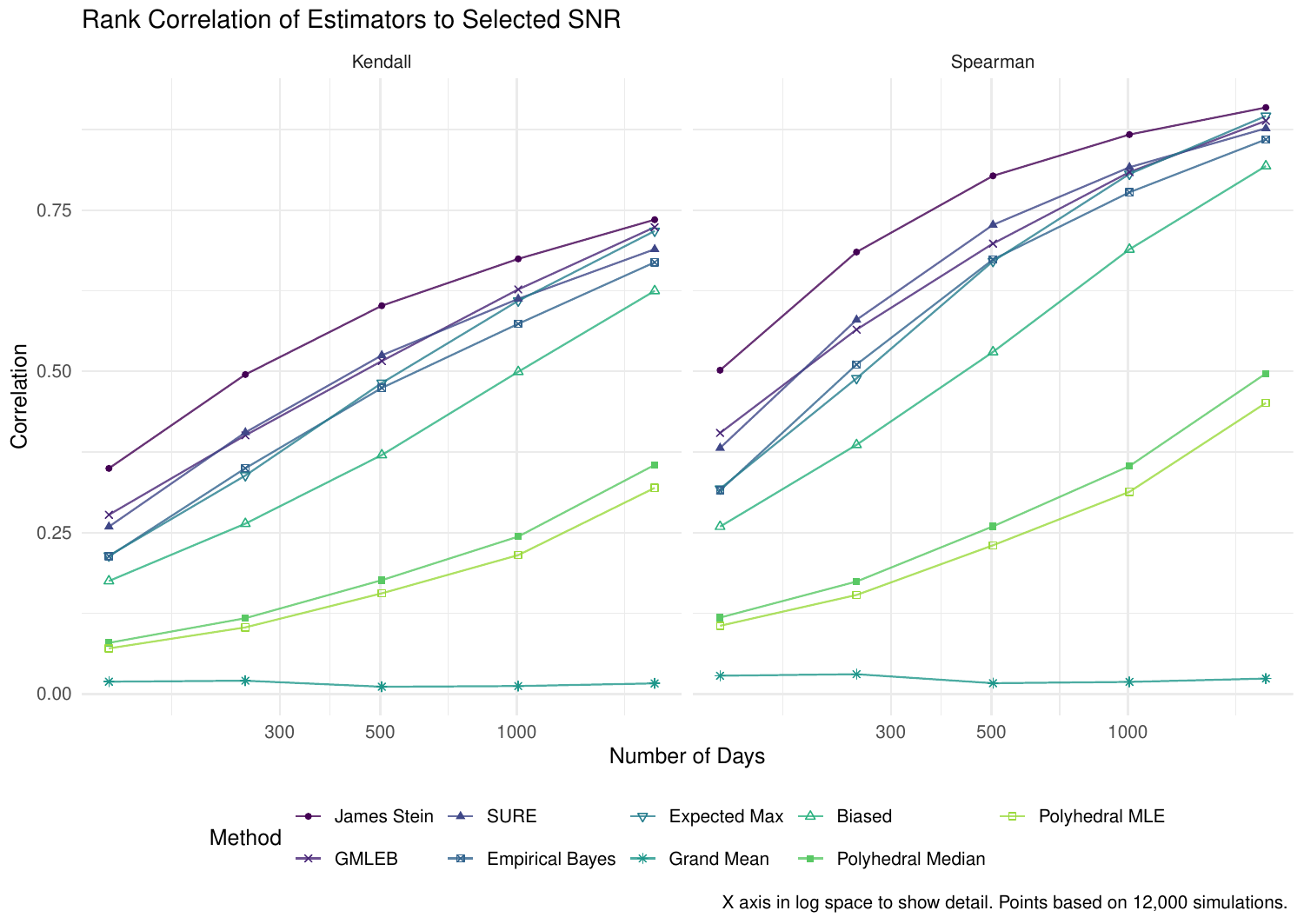} \caption[Rank correlation coefficients of the various tested estimators are plotted versus \ssiz, across all layouts and values of \nstrat and \psnrsig]{Rank correlation coefficients of the various tested estimators are plotted versus \ssiz, across all layouts and values of \nstrat and \psnrsig. Returns are independent. }\label{fig:anb_plot_III}
\end{figure}

\end{knitrout}
\begin{knitrout}\small
\definecolor{shadecolor}{rgb}{0.969, 0.969, 0.969}\color{fgcolor}\begin{figure}[ht]
\includegraphics[width=0.975\textwidth,height=0.691\textwidth]{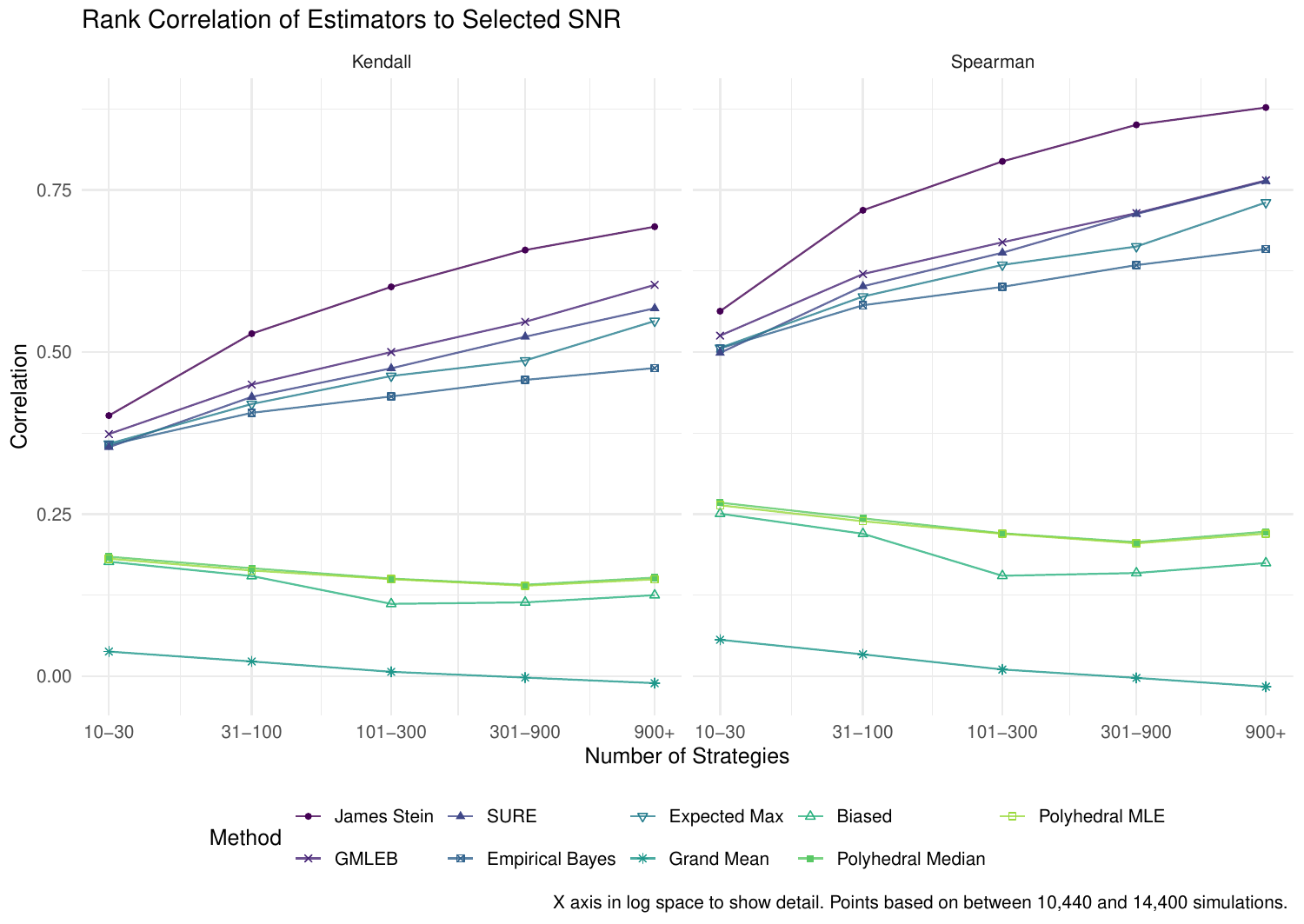} \caption[Rank correlation coefficients of the various tested estimators are plotted versus groups of \nstrat, across all layouts and values of \ssiz and \psnrsig]{Rank correlation coefficients of the various tested estimators are plotted versus groups of \nstrat, across all layouts and values of \ssiz and \psnrsig. Returns are independent. }\label{fig:anb_plot_IV}
\end{figure}

\end{knitrout}
\begin{knitrout}\small
\definecolor{shadecolor}{rgb}{0.969, 0.969, 0.969}\color{fgcolor}\begin{figure}[ht]
\includegraphics[width=0.975\textwidth,height=0.691\textwidth]{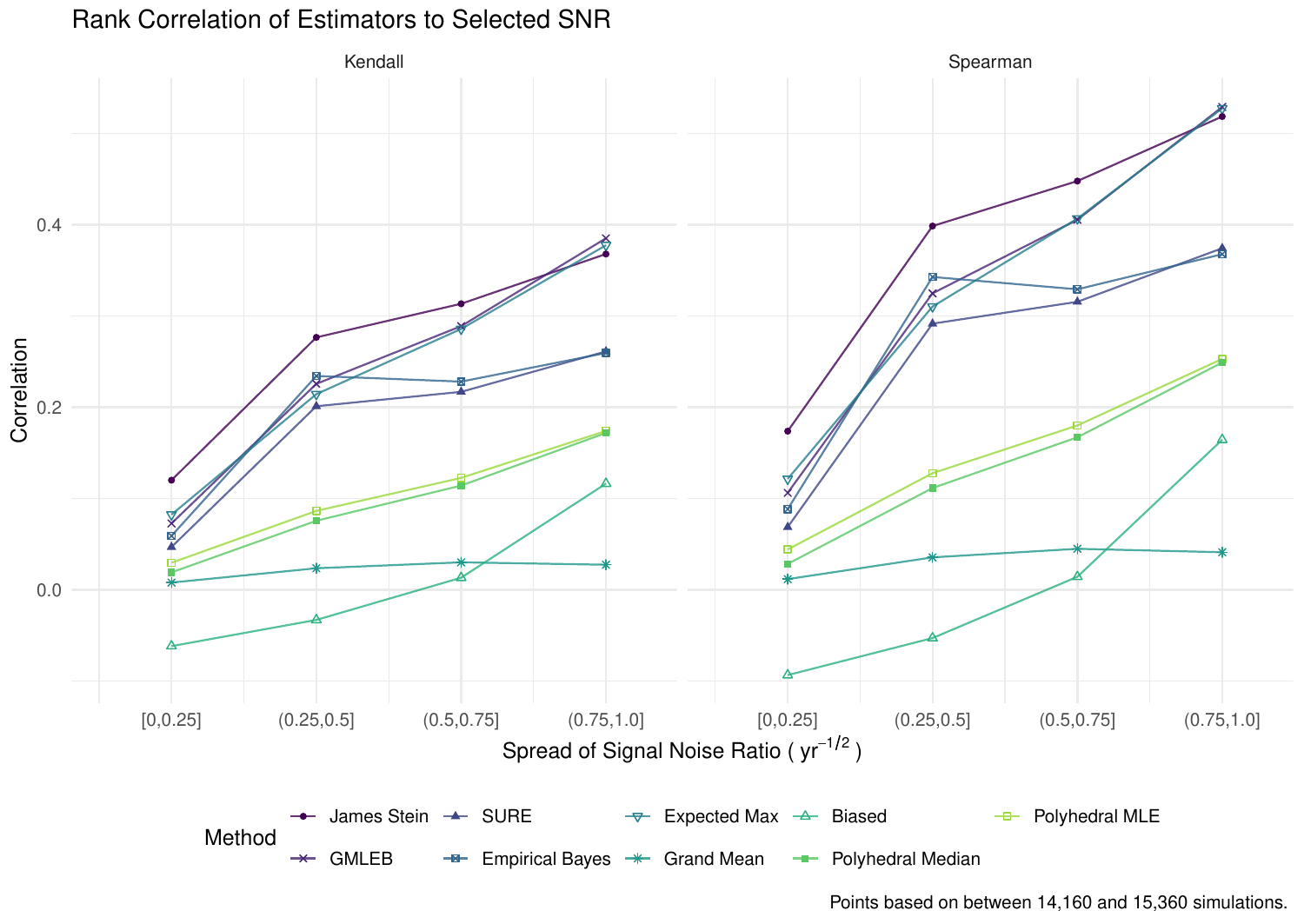} \caption[Rank correlation coefficients of the various tested estimators are plotted versus \psnrsig, across all layouts and values of \ssiz and \nstrat]{Rank correlation coefficients of the various tested estimators are plotted versus \psnrsig, across all layouts and values of \ssiz and \nstrat. Returns are independent. }\label{fig:anb_plot_V}
\end{figure}

\end{knitrout}

\clearpage

\section{Simulations Under Correlated Returns}

In the simulations considered heretofore in this note, the returns of assets were independent, that is $\RMAT=\eye$.
We wish to consider how these estimators will perform for more general correlation structures.
The simplest model, and one which is fairly accurate for a lot of quantitative work, is that of equicorrelation or compound symmetric structure
where $\RMAT = \wrapParens{1-\rho}\eye + \rho \vone\trvone$.
In our simulations we treat $\rho$ as another knob to be controlled.
We do not modify the estimators in any way to deal with correlation, rather we treat the correlation as a nuisance.

We note that in the extreme, for $\rho \approx 1.0$, the errors in \svsr are all perfectly aligned,
and the ordering of the \svsr will mirror that of \pvsnr. 
That is, with high probability $\psnr[\nstrat]$ will be the largest element of \pvsnr.
In this case there is no selection bias from selecting the asset based on the \txtSR,
and the Biased estimator should come to dominate.
We note that the estimators based on the polyhedral lemma have no actual dependence on the correlation structure, \cf
the additional analysis in \secref{polyhedral_lemma} in the appendix.
Thus we expect no real degradation of (otherwise poor) performance for the Polyhedral Median and Polyhedral MLE estimators.

Towards this end we perform simulations to estimate the bias and RMSE of the various estimators.
As above, we perform simulations using three layouts of \pvsnr.
We fix \psnrsig at $0.5\yrto{-\halff}$, 
$\ssiz=1008$, $\nstrat=100$.
We let $\rho$ vary from 
$0.05$ to 
$0.95$.
For each setting of the parameters we perform 500 simulations.

In \figref{sim_rnd_rho_plotz_bias} we plot the bias versus $\rho$;
in \figref{sim_rnd_rho_plotz_rmse} we plot the RMSE.
As perhaps expected, as $\rho \to 1$, the bias of all estimators becomes negative,
except the Biased estimator which becomes unbiased.
All the serious estimators are very biased for large $\rho$, though perhaps GMLEB has the
best performance for correlated asset returns.

In terms of RMSE, all estimators show increased RMSE as $\rho\to 1$, except the Biased estimator.
Likely much of this is due to the increased bias.
The Expected Max estimator has some of the worst performance in terms of RMSE for highly correlated
returns.

\begin{knitrout}\small
\definecolor{shadecolor}{rgb}{0.969, 0.969, 0.969}\color{fgcolor}\begin{figure}[ht]
\includegraphics[width=0.975\textwidth,height=0.691\textwidth]{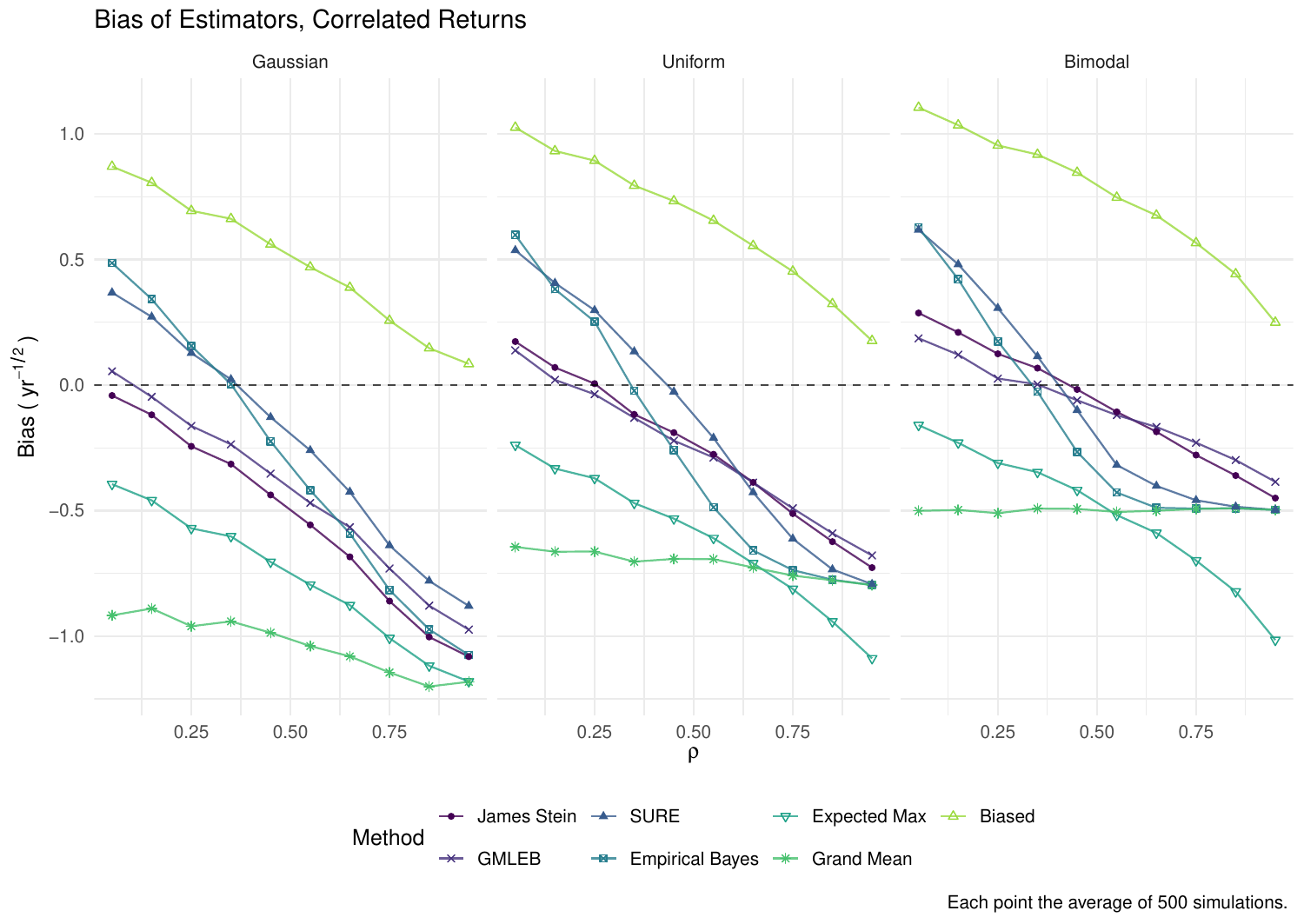} \caption[The empirical biases of the tested estimators are shown versus $\rho$, the common correlation of asset returns]{The empirical biases of the tested estimators are shown versus $\rho$, the common correlation of asset returns. Facet columns show the three different random layouts of the \pvsnr values. We fix $\ssiz=1008$, $\nstrat=100$, $\psnrsig=0.5\yrto{-\halff}$. }\label{fig:sim_rnd_rho_plotz_bias}
\end{figure}

\end{knitrout}

\begin{knitrout}\small
\definecolor{shadecolor}{rgb}{0.969, 0.969, 0.969}\color{fgcolor}\begin{figure}[ht]
\includegraphics[width=0.975\textwidth,height=0.691\textwidth]{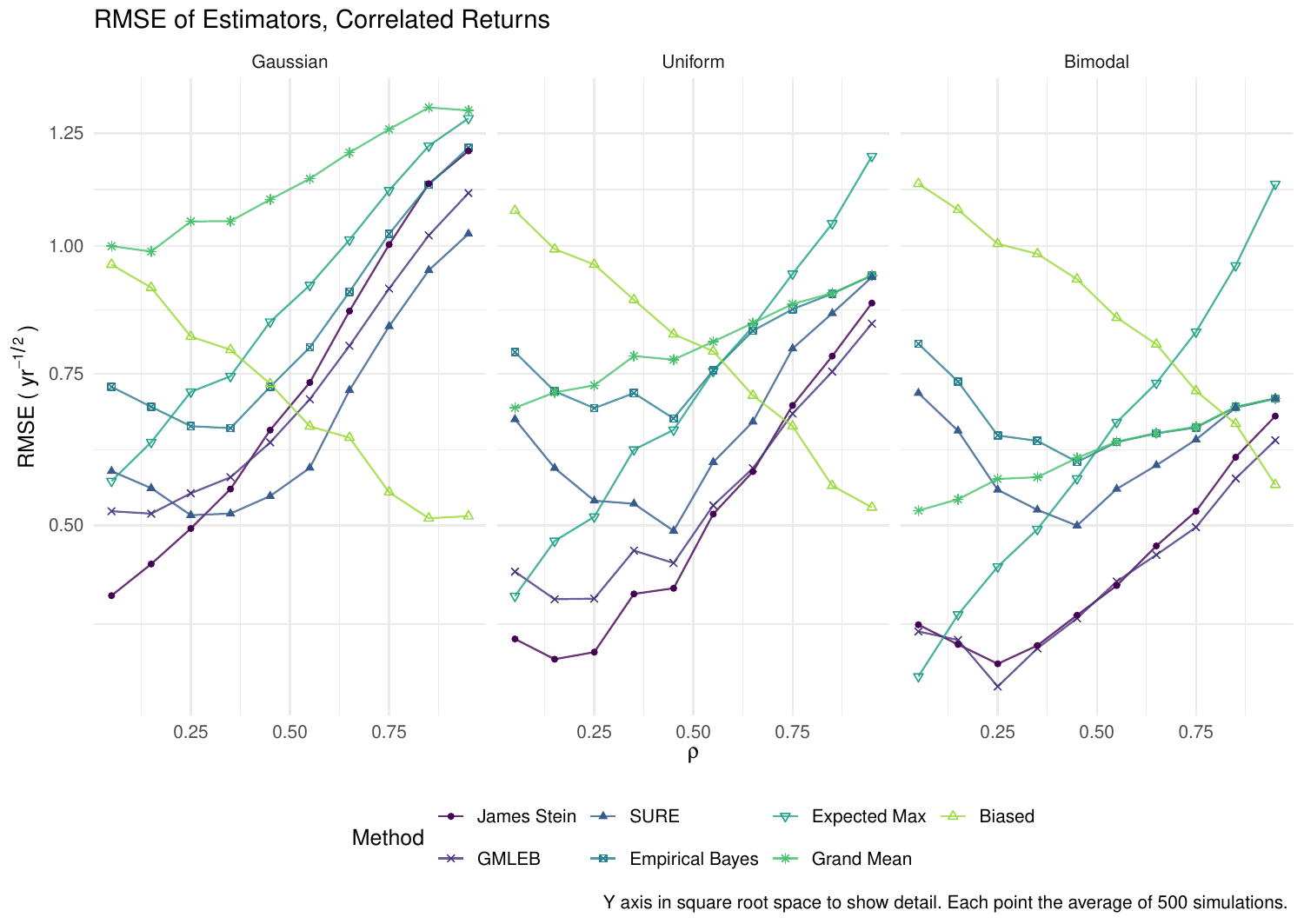} \caption[The empirical RMSE values of the tested estimators are shown versus $\rho$, the common correlation of asset returns]{The empirical RMSE values of the tested estimators are shown versus $\rho$, the common correlation of asset returns. Facet columns show the three different random layouts of the \pvsnr values. We fix $\ssiz=1008$, $\nstrat=100$, $\psnrsig=0.5\yrto{-\halff}$. }\label{fig:sim_rnd_rho_plotz_rmse}
\end{figure}

\end{knitrout}

\clearpage

In addition to the error simulations, 
we also ran ranking simulations, 
similar to those above, but with $\rho$ as an additional knob.
We created 800 corners, where 
we pick $\psnrsig$ uniformly from 
$0\yrto{-\halff}$ to
$1\yrto{-\halff}$,
we select \nstrat log uniformly from 
$10$ to
$2,500$,
and we select $\rho$ uniformly from 0 to 
$0.9$.
We cross this with \ssiz taking values from 
$0.5$ to $8$ years at $252$ days per year,
and cross this with the three different layouts, Gaussian, Uniform and Bimodal.
For each setting of the relevant parameters we performed 8 simulations.
This represents a total of 
96,000 simulations.

We compute all the estimators and the \psnr[\nstrat] in each simulation.
We then compute the rank correlation coefficients of each estimator to the ground truth,
sometimes grouping by relevant dimensions.

In \tabref{anbrho_1} we present the overall correlations over all simulations.
While James Stein estimator still dominates, and GMLEB and SURE traded places in the top three,
we see lower overall rank correlations than in the $\rho=0$ case
shown in \tabref{anb_1}.
That is, correlation of returns has caused degradation of performance of almost all the estimators.
The only exceptions are the Biased and Polyhedral Median estimators,
which is not surprising.

In \tabref{anbrho_2} we tabulate the rank correlations grouped by layout of the \pvsnr.
Somewhat surprisingly, James Stein is the top ranked method for Uniform and Bimodal distributions,
but is in a multi-way tie for second (or perhaps third) place for Gaussian returns,
losing out to the SURE estimator.
It is not clear what drives this result, since the James Stein dominates as an efficient estimator 
in the Gaussian case, but less so for Uniform and Bimodal layouts
as shown in \eg \figref{sim_rnd_plotz_rmse} for uncorrelated returns.


\begin{table}[ht]
\centering
\begin{tabular}{l|ll}
  \hline
Estimator & Kendall & Spearman \\ 
  \hline
James Stein & \textbf{0.34} & \textbf{0.48} \\ 
  SURE & 0.33 & 0.47 \\ 
  GMLEB & 0.32 & 0.45 \\ 
  Empirical Bayes & 0.31 & 0.45 \\ 
  Expected Max & 0.30 & 0.44 \\ 
  Polyhedral Median & 0.21 & 0.31 \\ 
  Biased & 0.19 & 0.27 \\ 
  Polyhedral MLE & 0.16 & 0.23 \\ 
  Grand Mean & 0.01 & 0.01 \\ 
   \hline
\end{tabular}
\caption{Rank correlation coefficients of the various tested estimators versus the \psnr[\nstrat] are shown for correlated returns. Rank correlations are computed across all layouts, \psnrsig, \ssiz, \nstrat, and $\rho$. Each estimator tested on 96,000 simulations. } 
\label{tab:anbrho_1}
\end{table}

\begin{table}[ht]
\centering
\begin{tabular}{ll|ll}
  \hline
Estimator & Layout & Kendall & Spearman \\ 
  \hline
SURE & Gaussian & \textbf{0.4} & \textbf{0.56} \\ 
  Empirical Bayes & Gaussian & 0.38 & 0.53 \\ 
  Expected Max & Gaussian & 0.38 & 0.53 \\ 
  James Stein & Gaussian & 0.38 & 0.52 \\ 
  GMLEB & Gaussian & 0.37 & 0.52 \\ 
  Polyhedral Median & Gaussian & 0.31 & 0.44 \\ 
  Polyhedral MLE & Gaussian & 0.25 & 0.36 \\ 
  Biased & Gaussian & 0.25 & 0.35 \\ 
  Grand Mean & Gaussian & 0.01 & 0.02 \\ 
   \hline
James Stein & Uniform & \textbf{0.36} & \textbf{0.5} \\ 
  SURE & Uniform & 0.33 & 0.47 \\ 
  Empirical Bayes & Uniform & 0.30 & 0.44 \\ 
  GMLEB & Uniform & 0.30 & 0.44 \\ 
  Expected Max & Uniform & 0.29 & 0.42 \\ 
  Polyhedral Median & Uniform & 0.18 & 0.26 \\ 
  Biased & Uniform & 0.16 & 0.23 \\ 
  Polyhedral MLE & Uniform & 0.12 & 0.17 \\ 
  Grand Mean & Uniform & 0.00 & 0.00 \\ 
   \hline
James Stein & Bimodal & \textbf{0.34} & \textbf{0.48} \\ 
  GMLEB & Bimodal & 0.29 & 0.42 \\ 
  SURE & Bimodal & 0.27 & 0.39 \\ 
  Empirical Bayes & Bimodal & 0.26 & 0.37 \\ 
  Expected Max & Bimodal & 0.23 & 0.33 \\ 
  Biased & Bimodal & 0.15 & 0.22 \\ 
  Polyhedral Median & Bimodal & 0.10 & 0.15 \\ 
  Polyhedral MLE & Bimodal & 0.05 & 0.07 \\ 
  Grand Mean & Bimodal & -0.00 & -0.00 \\ 
   \hline
\end{tabular}
\caption{Rank correlation coefficients of the various tested estimators versus the \psnr[\nstrat] are shown for correlated returns. Rank correlations are computed grouped by layout, across all \psnrsig, \ssiz, \nstrat, and $\rho$. Each estimator tested on 32,000 simulations. } 
\label{tab:anbrho_2}
\end{table}

In \figref{anbrho_plot_III}, we plot the rank correlation coefficients
against \ssiz across all values of \nstrat, \psnrsig and common correlation $\rho$.
We see that all estimators improve with increasing sample size, except the Grand Mean,
though some methods seem to be better suited for larger sample sizes.
Compared to \figref{anb_plot_III}, we see lower rank correlation coefficients overall, and perhaps
less distinction between the different serious estimators.

\begin{knitrout}\small
\definecolor{shadecolor}{rgb}{0.969, 0.969, 0.969}\color{fgcolor}\begin{figure}[ht]
\includegraphics[width=0.975\textwidth,height=0.691\textwidth]{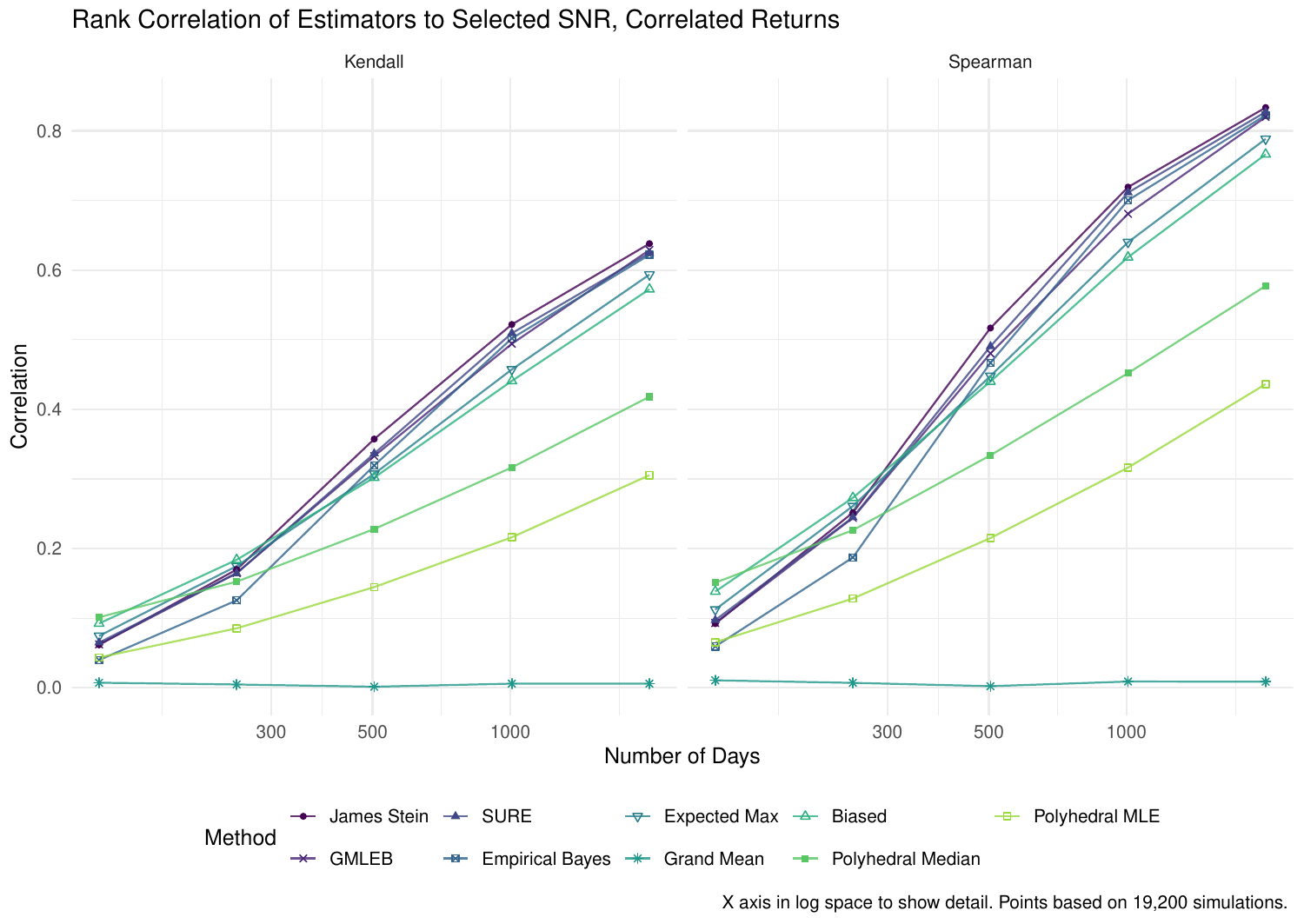} \caption[Rank correlation coefficients of the various tested estimators are plotted versus \ssiz for correlated returns]{Rank correlation coefficients of the various tested estimators are plotted versus \ssiz for correlated returns. Rank correlations are computed grouped by \ssiz, across all layouts, \nstrat, \psnrsig, and $\rho$. Rows based on 19,200 simulations. }\label{fig:anbrho_plot_III}
\end{figure}

\end{knitrout}

In \figref{anbrho_plot_IV}, we plot the rank correlations against cuts of \nstrat.
Compared to \figref{anb_plot_IV}, we see lower overall correlations,
but also the methods seem to peak and rank correlations get worse for larger \nstrat.
It is not clear what drives this phenomenon.
We see that James Stein estimator dominates for each grouping.
For purposes of comparing Bob and Alice when both backtested the same number of strategies,
the James Stein estimator seems to be uniformly the best choice.

\begin{knitrout}\small
\definecolor{shadecolor}{rgb}{0.969, 0.969, 0.969}\color{fgcolor}\begin{figure}[ht]
\includegraphics[width=0.975\textwidth,height=0.691\textwidth]{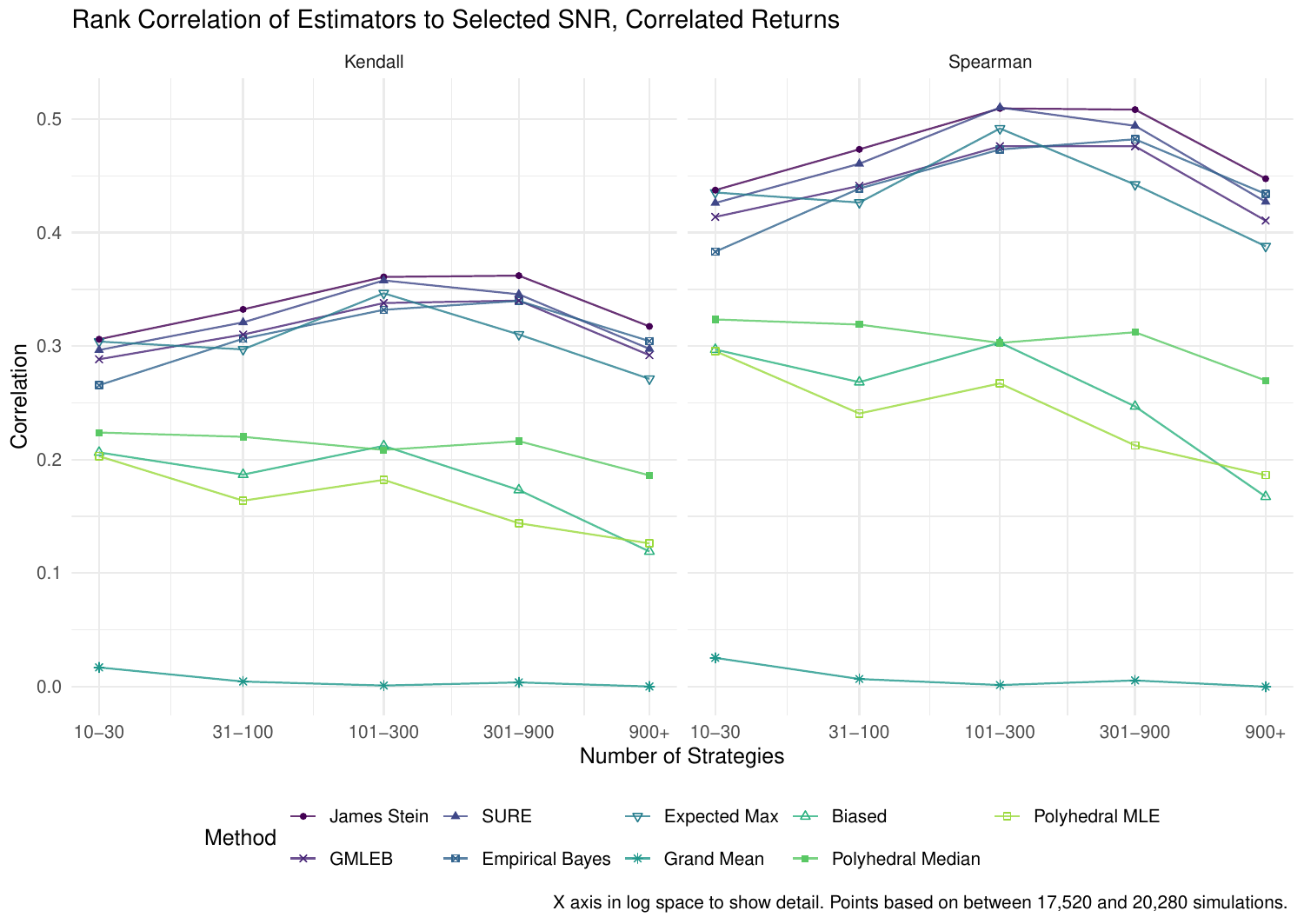} \caption[Rank correlation coefficients of the various tested estimators are plotted versus \nstrat for correlated returns]{Rank correlation coefficients of the various tested estimators are plotted versus \nstrat for correlated returns. Rank correlations are computed grouped by cuts of \nstrat, across all layouts, \ssiz, \psnrsig, and $\rho$. Rows based on between 17,520 and 20,280 simulations. }\label{fig:anbrho_plot_IV}
\end{figure}

\end{knitrout}

In \figref{anbrho_plot_V}, we plot the rank correlations against cuts of \psnrsig.
Compared to \figref{anb_plot_V}, we see lower overall correlations,

\begin{knitrout}\small
\definecolor{shadecolor}{rgb}{0.969, 0.969, 0.969}\color{fgcolor}\begin{figure}[ht]
\includegraphics[width=0.975\textwidth,height=0.691\textwidth]{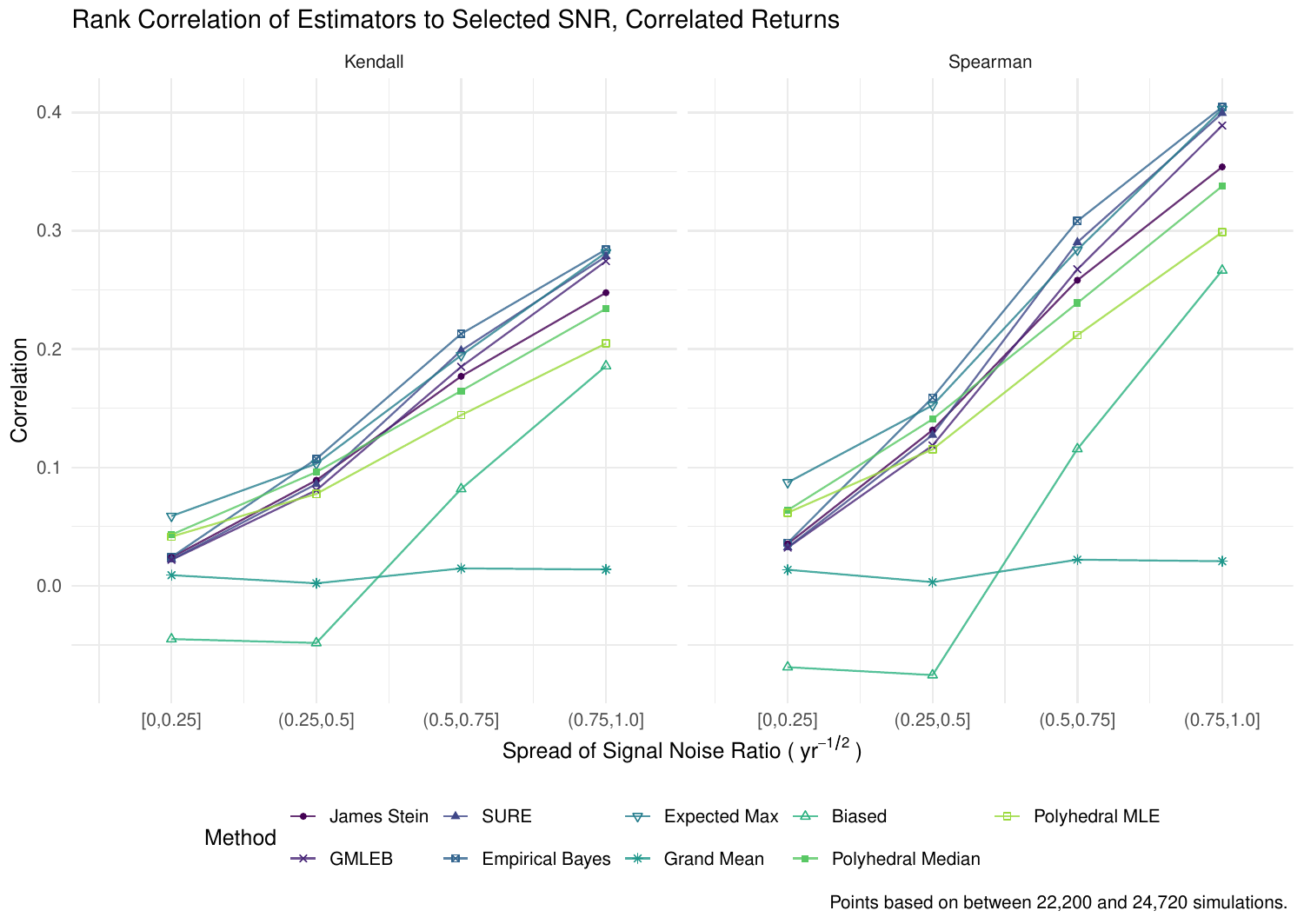} \caption[Rank correlation coefficients of the various tested estimators are plotted versus \psnrsig, called Spread, for correlated returns]{Rank correlation coefficients of the various tested estimators are plotted versus \psnrsig, called Spread, for correlated returns. Rank correlations are computed grouped by cuts of \psnrsig, across all layouts, \ssiz, \nstrat, and $\rho$. Rows based on between 22,200 and 24,720 simulations. }\label{fig:anbrho_plot_V}
\end{figure}

\end{knitrout}

In \tabref{anbrho_6}, we tabulate the rank correlation coefficients for the various methods
against cuts of the common correlation, $\rho$.
We plot the same in the companion plot, \figref{anbrho_plot_VI}.
Not surprisingly, given the discussion above and the findings in the accuracy studies,
the Biased estimator and those based on the polyhedral lemma show better performance
in the high $\rho$ case.
The other estimators suffer from increased correlation.
In particular, the James Stein estimator performs poorly when $\rho$ is fixed on very high values.

\begin{table}[ht]
\centering
\begingroup\scriptsize
\begin{tabular}{ll|ll}
  \hline
Estimator & Correlation & Kendall & Spearman \\ 
  \hline
James Stein & [0,0.25] & \textbf{0.49} & \textbf{0.68} \\ 
  GMLEB & [0,0.25] & 0.44 & 0.60 \\ 
  SURE & [0,0.25] & 0.44 & 0.61 \\ 
  Expected Max & [0,0.25] & 0.40 & 0.56 \\ 
  Empirical Bayes & [0,0.25] & 0.40 & 0.56 \\ 
  Polyhedral Median & [0,0.25] & 0.16 & 0.24 \\ 
  Polyhedral MLE & [0,0.25] & 0.16 & 0.24 \\ 
  Biased & [0,0.25] & 0.13 & 0.19 \\ 
  Grand Mean & [0,0.25] & 0.00 & 0.01 \\ 
   \hline
SURE & (0.25,0.5] & \textbf{0.4} & \textbf{0.57} \\ 
  James Stein & (0.25,0.5] & 0.38 & 0.54 \\ 
  GMLEB & (0.25,0.5] & 0.37 & 0.52 \\ 
  Empirical Bayes & (0.25,0.5] & 0.36 & 0.51 \\ 
  Expected Max & (0.25,0.5] & 0.36 & 0.50 \\ 
  Biased & (0.25,0.5] & 0.19 & 0.28 \\ 
  Polyhedral Median & (0.25,0.5] & 0.19 & 0.28 \\ 
  Polyhedral MLE & (0.25,0.5] & 0.17 & 0.25 \\ 
  Grand Mean & (0.25,0.5] & 0.01 & 0.01 \\ 
   \hline
SURE & (0.5,0.75] & \textbf{0.34} & \textbf{0.48} \\ 
  Expected Max & (0.5,0.75] & \textbf{0.34} & \textbf{0.48} \\ 
  Empirical Bayes & (0.5,0.75] & 0.31 & 0.44 \\ 
  James Stein & (0.5,0.75] & 0.30 & 0.43 \\ 
  GMLEB & (0.5,0.75] & 0.30 & 0.43 \\ 
  Biased & (0.5,0.75] & 0.28 & 0.40 \\ 
  Polyhedral Median & (0.5,0.75] & 0.24 & 0.35 \\ 
  Polyhedral MLE & (0.5,0.75] & 0.19 & 0.28 \\ 
  Grand Mean & (0.5,0.75] & 0.01 & 0.01 \\ 
   \hline
Biased & (0.75,0.9] & \textbf{0.32} & \textbf{0.46} \\ 
  Expected Max & (0.75,0.9] & 0.31 & 0.44 \\ 
  Polyhedral Median & (0.75,0.9] & 0.29 & 0.41 \\ 
  SURE & (0.75,0.9] & 0.25 & 0.37 \\ 
  Empirical Bayes & (0.75,0.9] & 0.23 & 0.33 \\ 
  James Stein & (0.75,0.9] & 0.22 & 0.32 \\ 
  GMLEB & (0.75,0.9] & 0.22 & 0.32 \\ 
  Polyhedral MLE & (0.75,0.9] & 0.21 & 0.31 \\ 
  Grand Mean & (0.75,0.9] & 0.01 & 0.01 \\ 
  \end{tabular}
\endgroup
\caption{Rank correlation coefficients of the various tested estimators versus the \psnr[\nstrat] are shown. Rank correlations are computed grouped by cuts of $\rho$ across all layouts, \psnrsig, \ssiz and \nstrat. Rows based on between 14,160 and 15,360 simulations. } 
\label{tab:anbrho_6}
\end{table}

\begin{knitrout}\small
\definecolor{shadecolor}{rgb}{0.969, 0.969, 0.969}\color{fgcolor}\begin{figure}[ht]
\includegraphics[width=0.975\textwidth,height=0.691\textwidth]{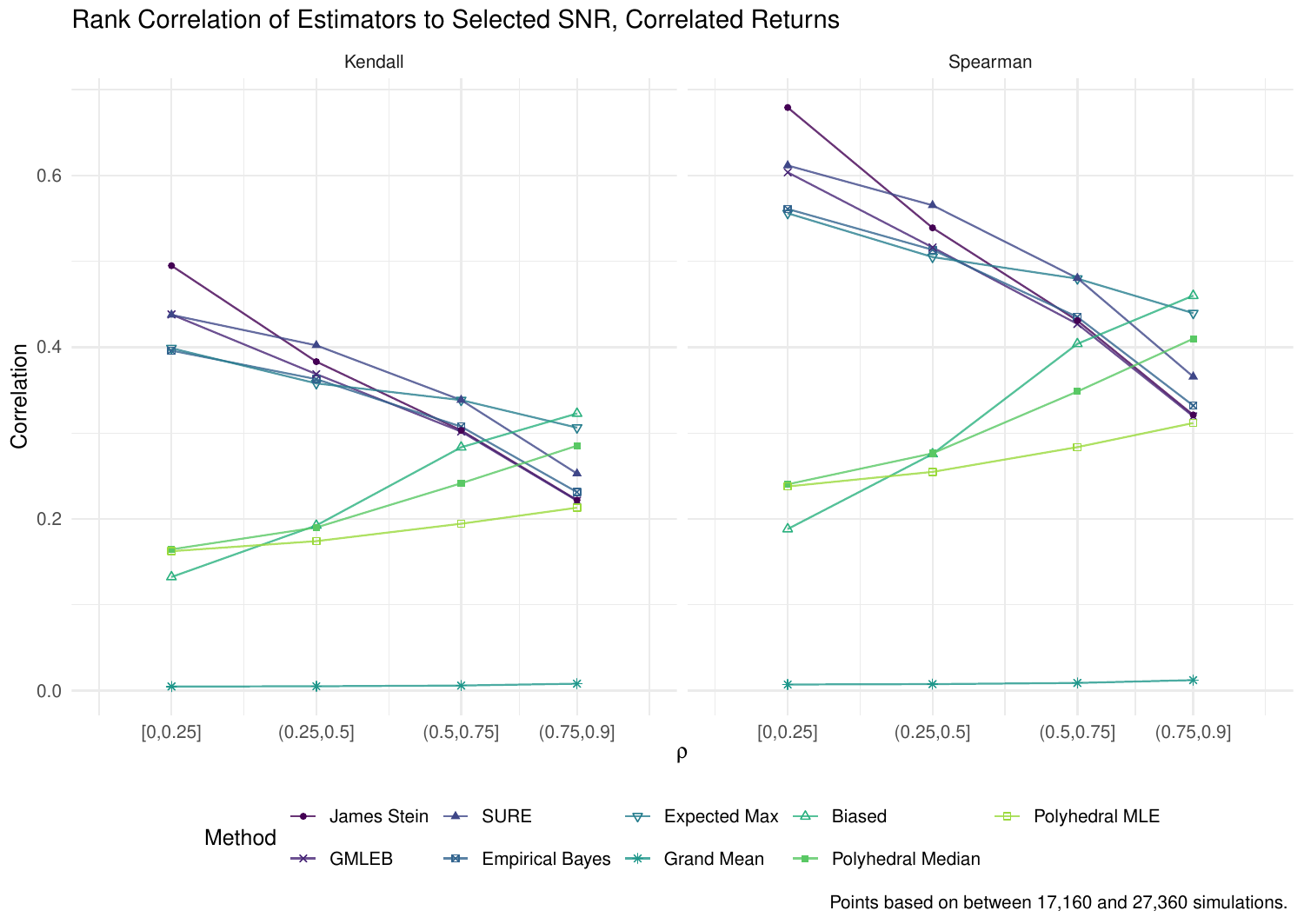} \caption[Rank correlation coefficients of the various tested estimators are plotted versus common correlation of \RMAT, $\rho$]{Rank correlation coefficients of the various tested estimators are plotted versus common correlation of \RMAT, $\rho$. Rank correlations are computed grouped by cuts of $\rho$, across all layouts, \ssiz, \nstrat, and \psnrsig. Rows based on between 17,160 and 27,360 simulations. }\label{fig:anbrho_plot_VI}
\end{figure}

\end{knitrout}


\clearpage

\section{Conclusions and Future Work}

Based on the bias, RMSE and correlation studies, we recommend the use of the James Stein estimator both as an estimator
and for comparing potential investment strategies.
The James Stein estimator is also recommended for its relative simplicity: it can be computed
knowing only \nstrat, \ssiz, \ssr[\nstrat], \ssravg and the empirical variance of the \svsr.
If one were somehow certain that there were only a few good strategies among all strategies tested,
we might recommend GMLEB, but the code is a bit more complicated.
We cannot recommend the estimators based on the polyhedral lemma, nor the Expected Max estimator, due to poor empirical performance.

These results are fairly robust against correlation of asset returns. 
If one needed a general purpose estimator for ranking selected strategies,
then James Stein should be used if the number of strategies \nstrat is always fixed,
or if \ssiz is fixed and fairly large.
If one were certain that returns were highly correlated with each other, however, 
the Biased estimator is recommended as it is effectively unbiased!

We consider this study still preliminary, as there are obvious improvements that could be made:
\begin{compactenum}
\item
In our analysis we viewed correlation of returns as a nuisance that threatens the robustness of our findings,
when in reality correlation is an unfortunate fact of life.
Some of the methods considered have obvious extensions to the case of correlated returns,
but a practical estimator would have to estimate the common correlation, or any other adjustments required of a more general \RMAT.
We have not implemented those procedures yet, but anticipate doing so in a future revision of this paper.


\item 
One can easily imagine some of the estimators could be improved by eliminating irrelevant elements of \svsr.
That is we would seek some way of choosing a threshold so that we ignore all tested strategies with \txtSR less than $\ssr[\nstrat] - \delta$
for some $\delta$ depending on \ssiz.
Perhaps such a threshold could be chosen using Hansen's log-log trick. \cite{pav2019maxsharpe,Hansen:2005}
In particular we suspect this might improve the less selective estimators like James Stein and Expected Max.
\item 
Many of the estimators considered here were constructed for the problem of reduced MSE estimation of the \emph{entire vector} \pvsnr,
not for the selective case of estimating $\psnr[\nstrat]$. 
Perhaps they can be further tuned to reduce MSE for our problem.
In particular it seems the SURE estimate could be improved perhaps by clever application of Stein's lemma to the truncated Gaussian conditional
distribution we get from the polyhedral lemma.
\item 
There is likely a way to look at \svsr and route it to one of the estimators considered here, or otherwise ensemble a few of the estimators
in a way that has even better RMSE.
\item
Estimators based on the polyhedral lemma are unstable when $\ssr[\nstrat] \approx \ssr[\nstrat-1]$. 
A better approach is needed for selective inference, one which recognizes that if $\ssr[\nstrat-1]$ were just a little bigger
we would be considering that strategy instead.
Such an approach would likely give tighter confidence intervals for the inferential problem as well. \cite{pav2019maxsharpe}
\item 
We plan to release the code for this paper as open source, after cleaning it up.
If you have an estimator you think is better, or you think we straw-manned your favorite estimator by our choice of testing parameters,
by all means re-run the code yourself.
\end{compactenum}


\bibliographystyle{plainnat}
\bibliography{common,debsr}

\clearpage
\appendix

\section{AI Use Statement}
We attempted to use AI in the preparation of this manuscript:

\begin{compactitem}
\item We asked gemini 3.0 for advice naming this paper. 
  This resulted in several awful suggestions, which we ignored.
  Ultimately, gemini suggested the ``of'' in the title, which we felt was better than the original, ``on.''
\item We asked gemini 3.0 several times to look over the code for our simulations.
  The LLM failed to find a critical error in an early iteration of our code.
\item We asked gemini 3.0 for help speeding up the GMLEB code.
  By vectorizing the computation, the LLM was able to achieve ballpark 10x speedups for large vectors
  while preserving accuracy compared to the reference implementation.
\item We asked various LLMs, including ChatGPT 5.2, luxor, minimax-m2.5, and trinity-large-preview, for help unearthing
  relevant references. Some of these models fabricated references, including fake DOI.
\item We asked gemini 3.0 to check for grammatical errors and typos. All remaining errors are the author's fault.
\end{compactitem}

\section{Polyhedral Lemma}
\label{sec:polyhedral_lemma}

Here we quote the polyhedral lemma, and then specialize it to the case of conditioning on the maximal element of a normally distributed vector.

\begin{theorem}[Lee \etal, Theorem 5.2 \cite{lee2013exact}]
\label{theorem:lee_etal}
Suppose $\yvec\sim\normlaw{\pvmu,\pvsig}$. 
  Define 
  $\ccc=\pvsig\etav / \qform{\pvsig}{\etav},$ and
  $\zzz=\yvec - \ccc\trAB{\etav}{\yvec}.$
  Let \pnorm[x] be the CDF of a standard normal, and 
  let \tncdf{x}{a}{b}{0}{1} be the CDF of a standard normal truncated
  to $\ccinterval{a}{b}$:
  $$
  \tncdf{x}{a}{b}{0}{1}\defeq\frac{\pnorm[x]-\pnorm[a]}{\pnorm[b] - \pnorm[a]}.
  $$
  Let \tncdf{x}{a}{b}{\pmu}{\psigsq} be the CDF of a general truncated normal,
  defined by
  $$
  \tncdf{x}{a}{b}{\pmu}{\psigsq} = 
  \tncdf{\frac{x-\pmu}{\psigma}}{\frac{a-\pmu}{\psigma}}{\frac{b-\pmu}{\psigma}}{0}{1}.
  $$
  Then, conditional on $\AAA\yvec\le\bbb$, the random variable
  $$
  \tncdf{\trAB{\etav}{\yvec}}{\Vmin}{\Vmax}{\trAB{\etav}{\pvmu}}{\qform{\pvsig}{\etav}}
  $$
  is Uniform on $\ccinterval{0}{1}$, where \Vmin and \Vmax are given by
  \begin{align*}
    \Vmin &= \max_{j:\wrapParens{\AAA\ccc}_j < 0} \frac{\bbb[j] -
    \wrapParens{\AAA\zzz}_j}{\wrapParens{\AAA\ccc}_j},\\
    \Vmax &= \min_{j:\wrapParens{\AAA\ccc}_j > 0} \frac{\bbb[j] -
    \wrapParens{\AAA\zzz}_j}{\wrapParens{\AAA\ccc}_j}.
  \end{align*}
\end{theorem}

In our case we condition on the maximal element of \yvec. 
This simplifies the result somewhat, especially in the case of independent returns:

\begin{corollary}
Suppose $\yvec\sim\normlaw{\pvmu,\pvsig}$, where \pvsig is diagonal.
Conditioning on $y_1 \le y_2 \le \ldots \le y_k$ then random variable
$$
\tncdf{y_k}{y_{k-1}}{\infty}{\pmu[k]}{\psigsq_{k,k}}
$$
is uniformly distributed on $\ccinterval{0}{1},$ 
where 
$\psigsq_{k,k}$ is the $k, k$ element of \pvsig.
\end{corollary}
\begin{proof}
Let \AAA be the \bby{(k-1)}{k} matrix which is mostly zeroes, but with a one on the principle diagonal and -1 on the next diagonal up:
$$
\AAA = \wrapBracks{\begin{array}{cccccc}
  1 & -1 & 0 & \ldots & 0 & 0 \\
  0 & 1 & -1 & \ldots & 0 & 0 \\
  0 & 0 & 1 &  \ldots & 0 & 0 \\
  \vdots & \vdots & \vdots & \ddots & \vdots & \vdots \\
  0 & 0 & 0 &  \ldots & 1 & -1
\end{array}},
$$
and let \bbb be the $k-1$ vector of all zeroes. 
Then we are conditioning on $\AAA\yvec\le\bbb$.

Take $\etav = \basev[k]$. 
Then $\pvsig\etav = \psigsq_{k,k}\basev[k]$,
$\qform{\pvsig}{\etav} = \psigsq_{k,k}$,
and $\trAB{\etav}{\pvmu} = \pmu[k]$.
In the language of the theorem, we then have
$\ccc=\basev[k]$ and 
$\zzz=\yvec - y_k\basev[k]$.
Then $\AAA\ccc$ and $\AAA\zzz$ are length $k-1$ vectors, and $\AAA\ccc$ is very sparse, with just a single non-zero element:
\begin{align*}
\AAA\ccc &= \tr{\wrapBracks{\begin{array}{ccccc} 0 & 0 & 0 & \ldots & -1 \end{array}}},\\
\AAA\zzz &= \tr{\wrapBracks{\begin{array}{ccccc} y_1 - y_2 & y_2 - y_3 & y_3 - y_4 & \ldots & y_{k-1} \end{array}}}.
\end{align*}
Then
$$
\wrapParens{\bbb - \AAA\zzz} \hadd \AAA\ccc  = \tr{\wrapBracks{\begin{array}{ccccc} \mathrm{NA} & \mathrm{NA} & \mathrm{NA} & \ldots & y_{k-1} \end{array}}}.
$$

Then there is only one element of $\AAA\ccc$ that is negative, and none that are positive. 
This establishes that $\Vmin = y_{k-1}$ and $\Vmax$ is the min over an empty set, or $\infty$.
Plugging these into the theorem, we get the desired result.
\end{proof}

Note that the fact that $y_k$ is the largest element of the \yvec had very little bearing on this derivation,
and in fact we could have been merely conditioning on $y_k \ge y_{k-1}$. 
If you modify \AAA to reflect that $y_k$ is the largest of the $y_i$, resulting in 
$$
\AAA = \wrapBracks{\begin{array}{cccccc}
  1 & 0 & 0 & \ldots & 0 & -1 \\
  0 & 1 & 0 & \ldots & 0 & -1 \\
  0 & 0 & 1 &  \ldots & 0 & -1 \\
  \vdots & \vdots & \vdots & \ddots & \vdots & \vdots \\
  0 & 0 & 0 &  \ldots & 1 & -1
\end{array}},
$$
the result is unchanged, assuming $y_{k-1}$ is the largest of the other elements.

Now consider how ill-posed the Polyhedral Median estimator is when $\ssr[\nstrat] = \ssr[\nstrat-1] + \epsilon$.
We wish to find $\psnr[0.5]$ such that
\begin{align*}
\half 
  &= \tncdf{\ssr[\nstrat]}{\ssr[\nstrat-1]}{\infty}{\pmu[0.5]}{\ssiz^{-1}},\\
  &= \frac{\pnorm[\sqrt{\ssiz}\wrapParens{\ssr[\nstrat] - \psnr[0.5]}] - 
\pnorm[\sqrt{\ssiz}\wrapParens{\ssr[\nstrat-1] - \psnr[0.5]}]}{1 - 
\pnorm[\sqrt{\ssiz}\wrapParens{\ssr[\nstrat-1] - \psnr[0.5]}]},\\
  &= \frac{\pnorm[\sqrt{\ssiz}\wrapParens{\ssr[\nstrat-1] + \epsilon - \psnr[0.5]}] - 
\pnorm[\sqrt{\ssiz}\wrapParens{\ssr[\nstrat-1] - \psnr[0.5]}]}{1 - 
\pnorm[\sqrt{\ssiz}\wrapParens{\ssr[\nstrat-1] - \psnr[0.5]}]},\\
  &\approx 
  \frac{\sqrt{\ssiz}\epsilon \dnorm[\sqrt{\ssiz}\wrapParens{\ssr[\nstrat-1] + \epsilon - \psnr[0.5]}]}{1 - 
\pnorm[\sqrt{\ssiz}\wrapParens{\ssr[\nstrat-1] - \psnr[0.5]}]}.
\end{align*}
When $\sqrt{\ssiz}\epsilon$ is small, this forces $1 - \pnorm[\sqrt{\ssiz}\wrapParens{\ssr[\nstrat-1] - \psnr[0.5]}]$ to be small as well,
which causes $\wrapParens{\ssr[\nstrat-1] - \psnr[0.5]}$ to be large, 
driving $\psnr[0.5] \to - \infty$.

Now consider the case of compound symmetric correlation among asset returns. 
The simplified form of the polyhedral lemma is as follows:

\begin{corollary}
Suppose $\yvec\sim\normlaw{\pvmu,\pvsig}$, where 
$\pvsig = \psigsq\wrapParens{\wrapParens{1-\rho} \eye + \rho\vone\trvone}$ for $\rho \ge 0$.
That is, elements of \yvec have common variance and equicorrelation of $\rho$.
Conditioning on $y_1 \le y_2 \le \ldots \le y_k$ then random variable
$$
  \tncdf{y_k}{\frac{y_{k-1} - \rho y_k}{1 - \rho}}{\infty}{\pmu[k]}{\psigsq}
$$
is uniformly distributed on $\ccinterval{0}{1}.$ 
\end{corollary}
\begin{proof}
Let \AAA be the \bby{(k-1)}{k} matrix which is mostly zeroes, but with a one on the principle diagonal and -1 on the next diagonal up:
$$
\AAA = \wrapBracks{\begin{array}{cccccc}
  1 & -1 & 0 & \ldots & 0 & 0 \\
  0 & 1 & -1 & \ldots & 0 & 0 \\
  0 & 0 & 1 &  \ldots & 0 & 0 \\
  \vdots & \vdots & \vdots & \ddots & \vdots & \vdots \\
  0 & 0 & 0 &  \ldots & 1 & -1
\end{array}},
$$
and let \bbb be the $k-1$ vector of all zeroes. 
Then we are conditioning on $\AAA\yvec\le\bbb$.

Take $\etav = \basev[k]$. 
Then 
$$\pvsig\etav = \psigsq\wrapParens{\rho\vone + \wrapParens{1-\rho}\basev[k]},$$
$\qform{\pvsig}{\etav} = \psigsq$,
and $\trAB{\etav}{\pvmu} = \pmu[k]$.
In the language of the theorem, we then have
$\ccc=\wrapParens{\rho\vone + \wrapParens{1-\rho}\basev[k]},$ and
$$
  \zzz=\yvec - y_k\rho\vone - y_k\wrapParens{1-\rho}\basev[k] = \wrapBracks{\begin{array}{c} y_1 - \rho y_k\\y_2 - \rho y_k\\ \vdots \\ y_{k-1} - \rho y_k \\ 0 \end{array}}.
$$
Then $\AAA\ccc$ and $\AAA\zzz$ are length $k-1$ vectors which take values:
\begin{align*}
\AAA\ccc &= \tr{\wrapBracks{\begin{array}{ccccc} 0 & 0 & 0 & \ldots & \rho -1 \end{array}}},\\
\AAA\zzz &= \tr{\wrapBracks{\begin{array}{ccccc} y_1 - y_2 & y_2 - y_3 & y_3 - y_4 & \ldots & y_{k-1} - \rho y_k \end{array}}}.
\end{align*}
Now there is only one element of $\AAA\ccc$ which is negative, and indeed the only that is non-zero, which is the last one.
This establishes that
$$
\Vmin = \frac{0 - \wrapParens{y_{k-1} - \rho y_k}}{\rho - 1} = \frac{y_{k-1} - \rho y_k}{1 - \rho}.
$$
Because $\AAA\ccc$ has no positive elements, $\Vmax=\infty$.
Plugging these into the theorem, we get the desired result.
\end{proof}


\section{Additional Tables and Figures}
\label{sec:additional_tables_and_figures}

\subsection{Ranking Results}

\begin{table}[ht]
\centering
\begingroup\scriptsize
\begin{tabular}{ll|ll}
  \hline
Estimator & Days & Kendall & Spearman \\ 
  \hline
James Stein & 126 & \textbf{0.35} & \textbf{0.5} \\ 
  GMLEB & 126 & 0.28 & 0.40 \\ 
  SURE & 126 & 0.26 & 0.38 \\ 
  Expected Max & 126 & 0.21 & 0.32 \\ 
  Empirical Bayes & 126 & 0.21 & 0.32 \\ 
  Biased & 126 & 0.18 & 0.26 \\ 
  Polyhedral Median & 126 & 0.08 & 0.12 \\ 
  Polyhedral MLE & 126 & 0.07 & 0.11 \\ 
  Grand Mean & 126 & 0.02 & 0.03 \\ 
   \hline
James Stein & 252 & \textbf{0.5} & \textbf{0.69} \\ 
  SURE & 252 & 0.41 & 0.58 \\ 
  GMLEB & 252 & 0.40 & 0.56 \\ 
  Empirical Bayes & 252 & 0.35 & 0.51 \\ 
  Expected Max & 252 & 0.34 & 0.49 \\ 
  Biased & 252 & 0.26 & 0.39 \\ 
  Polyhedral Median & 252 & 0.12 & 0.17 \\ 
  Polyhedral MLE & 252 & 0.10 & 0.15 \\ 
  Grand Mean & 252 & 0.02 & 0.03 \\ 
   \hline
James Stein & 504 & \textbf{0.6} & \textbf{0.8} \\ 
  SURE & 504 & 0.52 & 0.73 \\ 
  GMLEB & 504 & 0.52 & 0.70 \\ 
  Expected Max & 504 & 0.48 & 0.67 \\ 
  Empirical Bayes & 504 & 0.47 & 0.67 \\ 
  Biased & 504 & 0.37 & 0.53 \\ 
  Polyhedral Median & 504 & 0.18 & 0.26 \\ 
  Polyhedral MLE & 504 & 0.16 & 0.23 \\ 
  Grand Mean & 504 & 0.01 & 0.02 \\ 
   \hline
James Stein & 1008 & \textbf{0.67} & \textbf{0.87} \\ 
  GMLEB & 1008 & 0.63 & 0.81 \\ 
  SURE & 1008 & 0.61 & 0.82 \\ 
  Expected Max & 1008 & 0.61 & 0.81 \\ 
  Empirical Bayes & 1008 & 0.57 & 0.78 \\ 
  Biased & 1008 & 0.50 & 0.69 \\ 
  Polyhedral Median & 1008 & 0.24 & 0.35 \\ 
  Polyhedral MLE & 1008 & 0.22 & 0.31 \\ 
  Grand Mean & 1008 & 0.01 & 0.02 \\ 
   \hline
James Stein & 2016 & \textbf{0.74} & \textbf{0.91} \\ 
  GMLEB & 2016 & 0.72 & 0.89 \\ 
  Expected Max & 2016 & 0.72 & 0.90 \\ 
  SURE & 2016 & 0.69 & 0.88 \\ 
  Empirical Bayes & 2016 & 0.67 & 0.86 \\ 
  Biased & 2016 & 0.62 & 0.82 \\ 
  Polyhedral Median & 2016 & 0.36 & 0.50 \\ 
  Polyhedral MLE & 2016 & 0.32 & 0.45 \\ 
  Grand Mean & 2016 & 0.02 & 0.02 \\ 
   \hline
\end{tabular}
\endgroup
\caption{Rank correlation coefficients of the various tested estimators versus the \psnr[\nstrat] are shown. Rank correlations are computed grouped by \ssiz in days, across all layouts, \psnrsig, and \nstrat. Each estimator tested on 12,000 simulations. } 
\label{tab:anb_3}
\end{table}

\begin{table}[ht]
\centering
\begingroup\scriptsize
\begin{tabular}{ll|ll}
  \hline
Estimator & Strategies & Kendall & Spearman \\ 
  \hline
James Stein & 10-30 & \textbf{0.4} & \textbf{0.56} \\ 
  GMLEB & 10-30 & 0.37 & 0.53 \\ 
  Expected Max & 10-30 & 0.36 & 0.51 \\ 
  Empirical Bayes & 10-30 & 0.36 & 0.51 \\ 
  SURE & 10-30 & 0.35 & 0.50 \\ 
  Polyhedral Median & 10-30 & 0.18 & 0.27 \\ 
  Polyhedral MLE & 10-30 & 0.18 & 0.26 \\ 
  Biased & 10-30 & 0.18 & 0.25 \\ 
  Grand Mean & 10-30 & 0.04 & 0.06 \\ 
   \hline
James Stein & 31-100 & \textbf{0.53} & \textbf{0.72} \\ 
  GMLEB & 31-100 & 0.45 & 0.62 \\ 
  SURE & 31-100 & 0.43 & 0.60 \\ 
  Expected Max & 31-100 & 0.42 & 0.59 \\ 
  Empirical Bayes & 31-100 & 0.41 & 0.57 \\ 
  Polyhedral Median & 31-100 & 0.17 & 0.24 \\ 
  Polyhedral MLE & 31-100 & 0.16 & 0.24 \\ 
  Biased & 31-100 & 0.15 & 0.22 \\ 
  Grand Mean & 31-100 & 0.02 & 0.03 \\ 
   \hline
James Stein & 101-300 & \textbf{0.6} & \textbf{0.79} \\ 
  GMLEB & 101-300 & 0.50 & 0.67 \\ 
  SURE & 101-300 & 0.48 & 0.65 \\ 
  Expected Max & 101-300 & 0.46 & 0.63 \\ 
  Empirical Bayes & 101-300 & 0.43 & 0.60 \\ 
  Polyhedral Median & 101-300 & 0.15 & 0.22 \\ 
  Polyhedral MLE & 101-300 & 0.15 & 0.22 \\ 
  Biased & 101-300 & 0.11 & 0.15 \\ 
  Grand Mean & 101-300 & 0.01 & 0.01 \\ 
   \hline
James Stein & 301-900 & \textbf{0.66} & \textbf{0.85} \\ 
  GMLEB & 301-900 & 0.55 & 0.71 \\ 
  SURE & 301-900 & 0.52 & 0.71 \\ 
  Expected Max & 301-900 & 0.49 & 0.66 \\ 
  Empirical Bayes & 301-900 & 0.46 & 0.63 \\ 
  Polyhedral Median & 301-900 & 0.14 & 0.21 \\ 
  Polyhedral MLE & 301-900 & 0.14 & 0.21 \\ 
  Biased & 301-900 & 0.11 & 0.16 \\ 
  Grand Mean & 301-900 & -0.00 & -0.00 \\ 
   \hline
James Stein & 900+ & \textbf{0.69} & \textbf{0.88} \\ 
  GMLEB & 900+ & 0.60 & 0.76 \\ 
  SURE & 900+ & 0.57 & 0.76 \\ 
  Expected Max & 900+ & 0.55 & 0.73 \\ 
  Empirical Bayes & 900+ & 0.48 & 0.66 \\ 
  Polyhedral Median & 900+ & 0.15 & 0.22 \\ 
  Polyhedral MLE & 900+ & 0.15 & 0.22 \\ 
  Biased & 900+ & 0.13 & 0.17 \\ 
  Grand Mean & 900+ & -0.01 & -0.02 \\ 
   \hline
\end{tabular}
\endgroup
\caption{Rank correlation coefficients of the various tested estimators versus the \psnr[\nstrat] are shown. Rank correlations are computed grouped by cuts of \nstrat, across all layouts, \ssiz and \psnrsig. Rows based on between 10,440 and 14,400 simulations. } 
\label{tab:anb_4}
\end{table}

\begin{table}[ht]
\centering
\begingroup\scriptsize
\begin{tabular}{ll|ll}
  \hline
Estimator & Spread ($\yrto{-\halff}$) & Kendall & Spearman \\ 
  \hline
James Stein & [0,0.25] & \textbf{0.12} & \textbf{0.17} \\ 
  Expected Max & [0,0.25] & 0.08 & 0.12 \\ 
  GMLEB & [0,0.25] & 0.07 & 0.11 \\ 
  Empirical Bayes & [0,0.25] & 0.06 & 0.09 \\ 
  SURE & [0,0.25] & 0.05 & 0.07 \\ 
  Polyhedral MLE & [0,0.25] & 0.03 & 0.04 \\ 
  Polyhedral Median & [0,0.25] & 0.02 & 0.03 \\ 
  Grand Mean & [0,0.25] & 0.01 & 0.01 \\ 
  Biased & [0,0.25] & -0.06 & -0.09 \\ 
   \hline
James Stein & (0.25,0.5] & \textbf{0.28} & \textbf{0.4} \\ 
  Empirical Bayes & (0.25,0.5] & 0.23 & 0.34 \\ 
  GMLEB & (0.25,0.5] & 0.23 & 0.32 \\ 
  Expected Max & (0.25,0.5] & 0.21 & 0.31 \\ 
  SURE & (0.25,0.5] & 0.20 & 0.29 \\ 
  Polyhedral MLE & (0.25,0.5] & 0.09 & 0.13 \\ 
  Polyhedral Median & (0.25,0.5] & 0.08 & 0.11 \\ 
  Grand Mean & (0.25,0.5] & 0.02 & 0.04 \\ 
  Biased & (0.25,0.5] & -0.03 & -0.05 \\ 
   \hline
James Stein & (0.5,0.75] & \textbf{0.31} & \textbf{0.45} \\ 
  GMLEB & (0.5,0.75] & 0.29 & 0.41 \\ 
  Expected Max & (0.5,0.75] & 0.29 & 0.41 \\ 
  Empirical Bayes & (0.5,0.75] & 0.23 & 0.33 \\ 
  SURE & (0.5,0.75] & 0.22 & 0.32 \\ 
  Polyhedral MLE & (0.5,0.75] & 0.12 & 0.18 \\ 
  Polyhedral Median & (0.5,0.75] & 0.11 & 0.17 \\ 
  Grand Mean & (0.5,0.75] & 0.03 & 0.04 \\ 
  Biased & (0.5,0.75] & 0.01 & 0.01 \\ 
   \hline
GMLEB & (0.75,1.0] & \textbf{0.39} & \textbf{0.53} \\ 
  Expected Max & (0.75,1.0] & 0.38 & \textbf{0.53} \\ 
  James Stein & (0.75,1.0] & 0.37 & 0.52 \\ 
  SURE & (0.75,1.0] & 0.26 & 0.37 \\ 
  Empirical Bayes & (0.75,1.0] & 0.26 & 0.37 \\ 
  Polyhedral MLE & (0.75,1.0] & 0.17 & 0.25 \\ 
  Polyhedral Median & (0.75,1.0] & 0.17 & 0.25 \\ 
  Biased & (0.75,1.0] & 0.12 & 0.16 \\ 
  Grand Mean & (0.75,1.0] & 0.03 & 0.04 \\ 
   \hline
\end{tabular}
\endgroup
\caption{Rank correlation coefficients of the various tested estimators versus the \psnr[\nstrat] are shown. Rank correlations are computed grouped by cuts of \psnrsig, called Spread, across all layouts, \ssiz and \nstrat. Rows based on between 14,160 and 15,360 simulations. } 
\label{tab:anb_5}
\end{table}


\end{document}